\documentclass[11pt,a4paper]{article}
\usepackage{fullpage}
\usepackage{amsmath,amssymb,amsthm}
\usepackage{graphicx,hyperref,xspace}
\usepackage[font=small]{caption}

\renewcommand{\th}{%
    \ifmmode
        ^\mathrm{th}%
    \else%
        \textsuperscript{th}\xspace%
    \fi%
}
\newcommand{\st}{%
    \ifmmode
        ^\mathrm{st}%
    \else%
        \textsuperscript{st}\xspace%
    \fi%
}
\newcommand{\rd}{%
    \ifmmode
        ^\mathrm{rd}%
    \else%
        \textsuperscript{rd}\xspace%
    \fi%
}
\newcommand{\nd}{%
    \ifmmode
        ^\mathrm{nd}%
    \else%
        \textsuperscript{nd}\xspace%
    \fi%
}

\newcommand{\gone}{\sqrt{\scalebox{.8}{$\frac{\pi}{2}$}}}
\newcommand{\epslin}{\epsilon_{\scalebox{.55}{$L$}}}
\newcommand{\Amap}{{\sf A}}
\newcommand{\Amapd}{{\sf A}}
\newcommand{\Demb}{\cl D_{\cl E}}
\newcommand{\pE}{{p_{\scalebox{.8}{$\scriptscriptstyle \cl E$}}}}
\newcommand{\Id}{\bs I}
\newcommand{\scp}[2]{\langle #1,\, #2 \rangle}
\newcommand{\compl}{{\tt c}}
\newcommand{\transp}{{^{_{\!\rm T}}}}
\newcommand{\Rbb}{\mathbb{R}}

\newtheorem{definition}{Definition}
\newtheorem{proposition}{Proposition}
\newtheorem{corollary}{Corollary}
\newtheorem{lemma}{Lemma}

\newcommand{\supp}{{\rm supp}\,}
\newcommand{\tinv}[1]{{\textstyle\frac{1}{#1}}}
\newcommand{\sign}{{\rm sign}\,}

\newcommand{\ud}{\mathrm{d}} 
\renewcommand{\leq}{\leqslant}
\renewcommand{\geq}{\geqslant}

\DeclareMathOperator{\rank}{rank}

\DeclareMathOperator{\tr}{tr}
\DeclareMathOperator{\Idf}{Id}

\newcommand{\cc}{\centering}
\newcommand{\bs}{\boldsymbol}
\newcommand{\bb}{\mathbb}
\newcommand{\cl}{\mathcal}

\newcommand{\ts}{\textstyle}
\newcommand{\ie}{\emph{i.e.}, }
\newcommand{\eg}{\emph{e.g.}, }

\newcommand{\iid}{%
    \ifmmode
        \mathrm{iid}%
    \else%
        iid\xspace%
    \fi%
}
\newcommand{\rv}{\mbox{r.v.}\xspace}
\newcommand{\whp}{\mbox{w.h.p.}\xspace}

\title{
Time for dithering: fast and quantized random embeddings\\ 
via the restricted isometry property}
\date{\today}

\author{
Laurent Jacques$^{+,*}$ and Valerio Cambareri\thanks{LJ and VC are with Image and Signal Processing Group
  (ISPGroup), ICTEAM/ELEN, Universit\'e catholique de Louvain (UCL).  E-mail:
\url{laurent.jacques@uclouvain.be},
\url{valerio.cambareri@uclouvain.be}. The  authors are funded by the
Belgian F.R.S.-FNRS. Part of this study  is funded by the project
\textsc{AlterSense} (MIS-FNRS). $^{+}$Corresponding author.}}

\begin{document}

\maketitle

\begin{abstract}
Recently, many works have focused on the characterization of
non-linear dimensionality reduction methods obtained by quantizing linear
embeddings, \eg to reach fast processing time, efficient 
data compression procedures, novel geometry-preserving embeddings or to estimate the information/bits
stored in this reduced data representation.
In this work, we prove that many linear maps known to respect
the restricted isometry property (RIP) can induce a quantized random embedding with controllable
multiplicative and additive distortions with respect to the pairwise
distances of the data points beings considered. In other words, linear matrices
having fast matrix-vector multiplication algorithms (\eg based on partial
Fourier ensembles or on the adjacency matrix of unbalanced
expanders) can be readily used in the definition of fast
quantized embeddings with small distortions.    
This implication is made possible by applying right after the linear map an additive and random
\emph{dither} that stabilizes the impact of the uniform scalar
quantization operator applied afterwards. 

For different categories of RIP matrices, \ie for different linear embeddings
of a metric space $(\cl K \subset \bb R^n, \ell_q)$ in $(\bb R^m,
\ell_p)$ with $p,q \geq 1$, we derive upper bounds on the additive
distortion induced by quantization, showing that it decays either when the embedding
dimension $m$ increases or when the distance of a pair of embedded
vectors in $\cl K$ decreases.  
Finally, we develop a novel \emph{bi-dithered}
quantization scheme, which allows for a reduced distortion that decreases when the embedding
dimension grows and independently of the considered pair of vectors. 
\end{abstract}

\section{Introduction}
\label{sec:introduction}

Since the advent of the Johnson and Lindenstraus lemma in 1984 and its
numerous extensions~\cite{johnson1984extensions,dasgupta1999elementary}, non-adaptive dimensionality
reduction techniques for high dimensional data obtained through
random constructions are
now ubiquitous in numerous fields ranging from data mining and
database management~\cite{datar2004locality} to machine
learning algorithms~\cite{weinberger2009feature}, numerical linear
algebra~\cite{MAL-048,sarlos2006improved}, signal processing and
compressive sensing (CS)~\cite{candes2005decoding,baraniuk2008simple,davenport2010signal}.   

In the quest for such efficient techniques, a sub-field of
research has grown to study non-linear embeddings
obtained by quantizing the output of random linear mappings, \ie by coding it over a finite or countable set of values. This is
particularly important for reducing the number of bits required to
store the image of such maps. 
For instance, combinations of random linear maps with 1-bit sign operators~\cite{boufounos20081,jacques2013robust,plan2014dimension}, universal quantization
\cite{boufounos2012universal}, dithered uniform quantization
\cite{jacques2015quantized,jacques2014error} or local sketching
methods~\cite{andoni2006near}
were proved to be either as efficient as
unquantized linear embeddings, or to allow for better encoding of the
signal set geometry~\cite{rahimi2007random,boufounos2015representation}. 

In this context, a difficult question is the design of
non-linear maps with controllable distortion, fast encoding complexity (\eg log-linear in the signal dimension) and providing embeddings of \emph{large}, possibly continuous sets in $\bb R^n$
such as the set of sparse vectors, the one of
low-rank matrices or more general
\emph{low-complexity} vector sets. To the best of our knowledge, a few fast
non-linear maps already exist in the case of 1-bit quantization (\eg
built on a fast circulant linear operator), but these are currently restricted to
the embedding of finite sets~\cite{oymak2016near,yu2015binary,yu2014circulant}. 

A priori, the lack of quantized random maps with fast
encoding schemes and provable embedding properties of general sets is rather
striking. There is indeed a large availability of fast linear random
embeddings, \eg as those developed since 2004 in the literature of CS to satisfy the celebrated Restricted Isometry Property
(RIP) that precisely controls these embeddings' distortion (see, \eg~\cite{foucart2013mathematical,rudelson2008sparse,
  rauhut2012restricted,romberg2009compressive} and Sec.~\ref{sec:discussion}).

This work aims to fill this gap by showing that for a \emph{dithered}
uniform scalar quantization, random linear maps known to satisfy
the RIP determine, with high probability, a non-linear embedding
simply obtained by quantizing their image. We show that this is possible for
all vector sets on which this RIP is known to hold, hence ensuring a certain \emph{inheritance} from this linear embedding in the quantized domain. 

In a nutshell, the proofs developed in this work are not technical and are all based on
the same architecture: the RIP allows us to focus on the embedding of
the image of a low-complexity vector set (obtained through the corresponding
linear map) into a quantized domain thanks to a randomly-dithered quantization.
This is made possible by \emph{softening} the
discontinuous distances evaluated in the quantized domain according to a mathematical machinery inspired by
\cite{plan2014dimension} in the case of 1-bit quantization and extended in~\cite{jacques2015small} to
dithered uniform scalar quantization. Actually, this softening allows us to
extend the concentration of quantized random maps on a finite covering of
the low-complexity vector set $\cl K$ to this whole set by a
continuity argument. 
   
As already determined in many previous works, a key aspect in our
study is the way distances\footnote{The notion of distance may be simply a \emph{dissimilarity} associated to a pre-metric, see~Sec.~\ref{Main-results}.} are measured between vectors in both the quantized
embedding space $\cl E$ and in the original vector set $\cl K$. In particular, denoting by $(\ell_p,\ell_q)$-RIP$(\cl K, \epslin)$ (with $p,q \geq 1$) the property of a matrix $\bs
\Phi$ to approximate the
$\ell_q$-norm of all vectors in $\cl K$ through the $\ell_p$-norm of their
projections up to a multiplicative distortion $\epslin >0$ (see
\eqref{eq:RIPpq} for the precise definition),
we will first show that measuring distances in $\cl
E$ with the $\ell_1$-norm allows us to build, with high probability,
quantized embeddings inherited from the $(\ell_1,\ell_q)$-RIP$(\cl
K-\cl K,\epslin)$ of their linear map. In this case, the additive distortion of this embedding due to quantization is arbitrarily small provided the embedding dimension $m$ is large with respect to the
dimension of $\cl K$ (as measured by the Kolmogorov
entropy \cite{oymak2015near,plan2014dimension}). Second, if the distance between vectors in $\cl E$ is
measured with a squared $\ell_2$-norm, as classically done for linear
embeddings based on the RIP, we prove that a similar embedding inherited from the $(\ell_2,\ell_2)$-RIP$(\cl
K-\cl K,\epslin)$ holds with high probability, but the additive
distortion now vanishes only for arbitrarily close pairs of
embedded vectors. Finally, we show that one can obtain a quantized embedding having arbitrarily small
additive distortion by inheriting the $(\ell_2,\ell_2)$-RIP$(\cl
K-\cl K,\epslin)$ of its linear part. This is achieved by designing a \emph{bi-dithered}
quantization procedure that implicitly doubles the embedding dimension, and by endowing the embedding space $\cl E$ with a particular pre-metric that measures the dissimilarity between vectors.    

The rest of this paper is structured as follows. We first end this
introduction by giving an overview of the useful mathematical
notations. Sec.~\ref{Main-results} presents the preliminary concepts of our theoretical framework, followed by the statements of the three main results in this paper. Sec.~\ref{sec:discussion} discusses the portability of those
results and their connection to the existing ``market'' of RIP
matrices, including those with fast/low-complexity vector encoding
schemes, followed by some perspectives, relevant connections with existing works, and open problems. Finally, Sec.~\ref{proof-of-main-prop},
Sec.~\ref{sec:proof-distorted-embedding} and Sec.~\ref{sec:bi-dith-quant} are
dedicated to proving our main results, with a few useful lemmata postponed to Appendices~\ref{sec:useful-lemmas} and~\ref{sec:proof-lemma-5}.

\paragraph*{Conventions:} We find it useful to start with the
conventions and notations used throughout this paper. We will denote vectors and matrices with bold symbols, \eg~$\bs \Phi \in \bb R^{m\times n}$
or~$\bs u \in \bb R^m$, while lowercase light letters
are associated to scalar values. The identity function and the identity matrix in~$\bb R^n$ read
$\Idf$ and $\Id_{n}$, respectively, while $\bs 1_n := (1,\,\cdots,1)^\transp\in \bb R^n$ is the
vector of ones. The~$i\th$ component of a vector (or of a
vector function)~$\bs u$ reads either~$u_i$ or~$(\bs u)_i$, while the
vector~$\bs u_i$ may refer to the~$i\th$ element of a set of
vectors. The set of indices in~$\Rbb^d$ is~$[d]:=\{1,\,\cdots,d\}$ and
for any~$\cl S \subset [d]$ of cardinality~$S = |\cl S|$,
$\bs u_{\cl S} \in \bb R^{|\cl S|}$ denotes the restriction of
$\bs u$ to~$\cl S$, while $\bs B_{\cl S}$ is the matrix obtained by restricting the columns
of~$\bs B \in \bb R^{d \times d}$ to those indexed by~$\cl S$. The complement of a set $\cl S$ reads $\cl S^\compl$. For
any~$p\geq 1$, the~$\ell_p$-norm of~$\bs u$ is
$\|\bs u\|_p = (\sum_i |u_i|^p)^{1/p}$
with~$\|\!\cdot\!\|=\|\!\cdot\!\|_2$, and we shall denote with $\cl D_{\ell^p_p}(\bs a,\bs a') := \tinv{d} \|\bs a -
\bs a'\|_p^p$ the averaged $p\th$ power of the $\ell_p$-distance.   
The~$(n-1)$-sphere in~$\Rbb^n$ in $\ell_p$ is
$\bb S_{\ell_p}^{n-1}=\{\bs x\in\Rbb^n: \|\bs x\|_p=1\}$ while the unit ball reads~$\bb B_{\ell_p}^{n}=\{\bs x\in\Rbb^n: \|\bs x\|\leq 1\}$. For $\ell_2$, we will write $\bb B^n = \bb B^n_{\ell_2}$ and $\bb S^{n-1} = \bb S_{\ell_2}^{n-1}$. By extension, $\bb B_{\ell_F}^{n_1 \times n_2}$ is the Frobenius unit ball of $n_1 \times n_2$ matrices $\bs U$ with $\|\bs U\|_F \leq 1$, while $\ell_F$ will denote the metric associated to the Frobenius norm ${\|\cdot\|_F}$. We recall that this norm is associated to the scalar product $\scp{\bs U}{\bs V}_F = \tr(\bs U^\top \bs V)$ through $\|\bs U\|_F^2 = \scp{\bs U}{\bs U}_F$, for two matrices $\bs U, \bs V$.

We will use the simplified notation~$\cl X^{m\times
  n}$ and~$\cl X^{m}$ to denote an~$M\times N$
random matrix or an~$m$-length random vector, respectively, whose entries are
\emph{identically and independently distributed} (or \iid) as the probability distribution
$\cl X$, \eg $\cl N^{m \times n}(0,1)$ (or $\cl U^m([0, \delta])$) is
the distribution of a matrix (resp. vector)
whose entries are iid as 
the standard normal distribution~$\cl N(0,1)$ (resp. the uniform
distribution~$\cl U([0, \delta])$). We also use extensively the sub-Gaussian and sub-exponential
characterization of random variables (or \rv) and of random vectors detailed in
\cite{vershynin2010introduction}. The sub-Gaussian and the
sub-exponential norms of a random
variable $X$ are thus denoted by $\|X\|_{\psi_2}$ and $\|X\|_{\psi_1}$, respectively,
with the Orlicz norm $\|X\|_{\psi_\alpha} := \sup_{p\geq 1} p^{-1/\alpha} (\bb E
|X|^p)^{1/p}$ for $\alpha \geq 1$. The random variable $X$ is therefore
sub-Gaussian (or sub-exponential) if $\|X\|_{\psi_2} < \infty$
(resp. $\|X\|_{\psi_1} < \infty$). We also define $\|X\|_{\infty} :=
\inf\{s > 0: \bb P[|X| \leq s] = 1\}$, with the useful relation
$\|X\|_{\psi_\alpha} \leq \|X\|_{\infty}$. Roughly speaking, we will
sometimes write that an event
holds \emph{with high probability} (or \whp) if its probability of
failure decays exponentially with some specific dimensions depending
on the context (\eg with the
embedding dimension of a map). The common flooring and ceiling operators
are denoted $\lfloor \cdot \rfloor$ and $\lceil \cdot \rceil$,
respectively, and the positive thresholding function is defined by
$(\lambda)_+ := \tinv{2}(\lambda + |\lambda|)$ for
any~$\lambda\in\Rbb$. Finally, an important feature of our study is that we do not pay any particular
attention to constants in the many bounds developed in this paper. For instance,
the symbols $C,C', C'', ..., c,c',c'',... > 0$ are positive constants
whose values can change from one line to the other. We will, however,
limit the use of such constants and find convenient the use of (non-asymptotic) ordering notations $A \lesssim B$ (or $A \gtrsim
B$), if
there exists a $c > 0$ such that $A \leq c B$ (resp. $A \geq c B$) for
two quantities $A$ and~$B$. Moreover, we will write $A \simeq B$ if we have both $A\simeq B$ and $B \simeq A$. This last notation should not be confused with the writing $\cl A \cong \cl B$ that means that the two spaces $\cl A$ and $\cl B$ are isomorphic, \eg $\bb R^{n_1 \times n_2} \cong \bb R^n$ with $n = n_1n_2$. 

\section{Main Results}
\label{Main-results}

The Restricted Isometry Property (RIP) is a well-known criterion for characterizing the distortion induced by a
linear map $\bs \Phi \in \bb R^{m \times n}$, with $m$ generally
smaller than $n$, when $\bs \Phi$ is used to
embed a ``low-complexity'' set $\cl K \subset \bb R^n$
of the metric space $(\bb R^n, \ell_q)$ into $(\bb R^m, \ell_p)$.   In
particular, given $\epslin \in (0,1)$, we say that $\bs \Phi$ respects
the $(\ell_p,\ell_q)$-Restricted Isometry Property over $\cl K$, or
$(\ell_p,\ell_q)$-RIP$(\cl K, \epslin)$ if, for all $\bs x \in \cl K$,
\begin{equation}
  \label{eq:RIPpq}
  (1- \epslin) \|\bs x\|^p_q \leq \tfrac{\mu_{\Phi}}{m} \|\bs \Phi \bs
  x\|^p_p \leq  (1 + \epslin) \|\bs x\|^p_q,
\end{equation}
for some $\mu_\Phi > 0$ independent of $n$ and $m$ but possibly
dependent on $\cl K$ and $\epslin$, as in the case $p=q=1+O(1)$ and
$\cl K = \Sigma_{s,n} := \{\bs u \in \bb R^n:
|\supp \bs u| \leq s\}$~\cite{berinde2008combining}. For simplicity, we
consider henceforth that $\mu_\Phi = 1$, up to a rescaling of
the sensing matrix $\bs \Phi$.  

The RIP can be shown to hold for random matrix constructions, \eg when the entries
of $\bs \Phi$ are \iid sub-Gaussian~\cite{krahmer2014unified} or when this matrix is
connected to the adjacency matrix of some unbalanced expander graphs
\cite{berinde2008combining}. Moreover, many random constructions of $\bs \Phi$ with fast
vector encoding time, \ie low-complexity matrix-vector multiplication in the case of partial random basis~\cite{foucart2013mathematical} (\eg Fourier or Hadamard bases), random
convolutions~\cite{rauhut2012restricted,romberg2009compressive}, and spread-spectrum
techniques~\cite{puy2012universal} are also well known in the case $p=q=2$. We remark that the embedding
distortion explained by this RIP is \emph{multiplicative}
and represented by the factors~$1\pm \epslin$. \medskip  

This work addresses the question of combining such RIP matrix
constructions with a \emph{quantization} procedure in order to reach
efficient non-linear embedding procedures, \ie with a controllable level
of distortion between the original distance of any pair of vectors in~$\bb R^n$
and some suitably defined distance (or dissimilarity measure) of their quantized projections. This happens to be
important for efficiently storing, transmitting or even processing the
vectors undergoing this non-linear embedding~\cite{boufounos2015quantization}.  While other works
have analyzed the important case of 1-bit quantization (as reached by a
sign operator)
\cite{boufounos20081,jacques2013robust,plan2014dimension,oymak2015near,yu2014circulant},
we here focus on a uniform (mid-rise) scalar quantization process
\begin{equation}
  \label{eq:quant-def}
\ts \cl Q:\ \lambda \in \bb R\ \mapsto\ \bb \delta (\lfloor \tfrac{\lambda}{\delta} \rfloor + \tinv{2}) \in \bb Z_\delta := \delta(\bb Z + \tinv{2})
\end{equation}
of resolution $\delta > 0$ that is applied componentwise onto vectors. 

It is then straightforward that, by using the deterministic relation $|\cl Q(\lambda) - \lambda| \leq
\delta/2$ and the $(\ell_p,\ell_q)$-RIP$(\cl K- \cl K, \epslin)$, we quickly arrive to 
\begin{equation}
  \label{eq:straightforward-quant-embed}
  (1 - \epslin)^{1/p} \|\bs x - \bs x'\|_q - \delta \leq \tfrac{1}{\sqrt[p]{m}} \|\cl Q(\bs \Phi \bs
  x) - \cl Q(\bs \Phi \bs
  x')\|_p \leq  (1 + \epslin)^{1/p} \|\bs x - \bs x'\|_q + \delta,  
\end{equation}
for all $\bs x, \bs x' \in \cl K$. However, this deterministic derivation seems to hint at a severe drawback:
the new \emph{additive} distortion $\pm \delta$ appearing on both sides of
\eqref{eq:straightforward-quant-embed} is constant, irrespectively of the
embedding dimension $m$. This seems rather counterintuitive since, for
instance, 
vectors $\bs x$ and $\bs x'$ that are \emph{consistent}, \ie $\cl Q(\bs \Phi \bs
  x) = \cl Q(\bs \Phi \bs
  x')$, will hardly be far apart when $m$ increases,
  \ie the set of constraints corresponding to this consistency shall only be satisfied for sufficiently close vectors, whereas~\eqref{eq:straightforward-quant-embed} seems to provide just $\|\bs x - \bs x'\|_q \leq (1 -
  \epslin)^{-1/p} \delta \simeq \delta$ which does not decrease with $m$.

Despite this intuition, the quantizer \emph{alone} cannot prevent a constant additive distortion for all vectors in $\cl K$ and any matrix $\bs \Phi$. Indeed, taking $\bs x' = - \bs x$ and assuming $\{\bs x, \bs x'\} \subset \cl K$, since $\cl Q$ is odd\footnote{$\cl Q$ is in fact odd almost everywhere so that, \eg for a random Gaussian matrix, $\bb P[\cl Q(\bs \Phi \bs x') = - \cl Q(\bs \Phi \bs x)] = 1$.} and $|\cl Q((\bs \Phi \bs x)_i)| > \delta/2$ for $i \in [m]$, we have $\tfrac{1}{\sqrt[p]{m}} \|\cl Q(\bs \Phi \bs x) - \cl Q(\bs \Phi \bs
  x')\|_p = \tfrac{2}{\sqrt[p]{m}} \|\cl Q(\bs \Phi \bs x)\|_p > \delta$. If $\cl K$ allows $\|\bs x - \bs x'\|_q = 2\|\bs x\|_q$ to be arbitrarily small, a variant of \eqref{eq:straightforward-quant-embed} with an additive distortion decaying either with $\|\bs x - \bs x'\|_q$ or when $m$ increases cannot hold.     

The objective of this work is to show that, thanks to a
key randomization of the quantizer's input through a common \emph{dithering} procedure
\cite{gray1998quantization,wannamaker1998theory}, the resulting quantized embeddings
display both a multiplicative and an additive distortion of the
distances between vectors in $\cl K$, just like
in~\eqref{eq:straightforward-quant-embed}. With this dithering the additive
distortion, which is not present in the RIP of linear embeddings, decays when $m$ increases and/or when the
pairwise vector distances decrease. More
specifically, our aim is to analyze the embedding properties of
the quantized map $\Amap: \bb R^n \to \bb Z^m_\delta $ such that,
for $\bs x \in \bb R^n$,  
\begin{equation}
  \label{eq:quant-embed-def}
  \Amap(\bs x) = \Amap(\bs x, \bs \Phi, \bs \xi):= \cl Q(\bs \Phi \bs x + \bs \xi),
\end{equation}
where $\bs \xi$ is the uniform dither whose entries are \iid as
$\cl U([0, \delta])$, \ie $\bs \xi \sim \cl U^m([0, \delta])$. 
Notice that such a dithered and quantized random map is not new and some of
its properties have already been studied, \eg in the context of  
locality-sensitive hashing of finite vector sets (for Gaussian or
$\alpha$-stable random matrix distributions)
\cite{andoni2006near,datar2004locality} or for universal encoding
strategies when $\cl Q$ is replaced by a \emph{non-regular}\footnote{A quantizer $\cl Q$ is non-regular if its quantization bins are non-convex, \eg if $\cl Q$ is periodic over $\bb R$~\cite{boufounos2012universal}.}
quantizer or a more general periodic function~\cite{boufounos2012universal,boufounos2015universal,boufounos2015representation}
(see Sec.~\ref{sec:comp-with-other}). 

The insertion of a dithering in \eqref{eq:quant-embed-def} is also motivated by the following key property: in
  expectation, \emph{dithering cancels out quantization}. Indeed, as shown in
  \cite[App. A]{jacques2015small}, for $u\sim \cl U([0,1])$ and
  $a,a'\in \bb R$, $\bb E_u |\lfloor a + u \rfloor - \lfloor a' + u
  \rfloor| = |a - a'|$ so that for $\bs x, \bs x' \in \bb R^n$ and $\bs \xi \sim \cl U^m([0,\delta])$ 
  \begin{equation}
    \label{eq:basic-dith-quant-prop}
\ts \bb E_{\bs \xi} |\Amap(\bs x)_i - \Amap(\bs x')_i| = \delta \bb E_{\bs \xi} | \lfloor
(\delta^{-1} \bs \Phi\bs x + \delta^{-1}\bs \xi)_i \rfloor - \lfloor
(\delta^{-1} \bs \Phi\bs x' + \delta^{-1}\bs \xi)_i \rfloor | = |\big(\bs \Phi (\bs x -
\bs x')\big)_i|.
\end{equation}
It seems thus intuitive that by accumulating and ``averaging'' $m$ observations
of $\bs x$ and $\bs x'$, each associated to some \iid dither, one
should be able to approach this expectation where only the properties of the linear embedding matter. The exact meaning of this phenomenon will be
carefully studied below.

In order to make this analysis more formal, the embedding space
$\cl E = \bb Z_\delta^{m}$ is endowed with a pre-metric $\Demb: \cl E
\times \cl E \to \bb R_+$, \ie a measure of the dissimilarity between vectors of $\cl E$. This pre-metric is assumed homogeneous of degree $\pE$, \ie for all $\lambda \in
\bb R$, $\Demb(\lambda \bs a,\lambda \bs a') = |\lambda|^\pE \Demb(\bs
a, \bs a')$ for all $\bs a,\bs a' \in \bb R^m$; for instance, we will consider\footnote{We can remark that the function $\cl
D_{\ell_p}$ is a pre-metric but
not a metric since it does not respect the triangle inequality for $p > 1$.} $\Demb = \cl
D_{\ell_{\pE}}$ with $\cl
D_{\ell_p}(\bs a, \bs a') := \tinv m \|\bs a - \bs a'\|_p^p$.   

Within this context, we are
thus interested in the following
characterization of the map~$\Amap$. 
\begin{definition}[Quantized RIP]
\label{def:QRIP}
For $\epslin,\epsilon \in (0,1)$, $q\geq 1$,
and some \emph{additive distortion} $\rho: (\epsilon, s)\in
\bb R^+\times \bb R^+ \mapsto \rho(\epsilon, s) \in \bb R^+$ such that $\rho$ is continuous,
\begin{equation}
  \label{eq:cond-rho-qrip}
\rho(0,0) = 0\text{ and }\lim_{s \to + \infty} s^{-\pE} \rho(\epsilon,
s) = 0,  
\end{equation}
a map $\Amap: \bb R^n \to \cl E$ respects the $(\Demb,\ell_q)$-\emph{quantized restricted
isometry property}, or $(\Demb,\ell_q)$-QRIP$(\cl K, \epslin,
\rho)$, if for all $\bs x, \bs x' \in \cl K$
\begin{equation}
  \label{eq:QRIP-def}
(1-\epslin) \|\bs x - \bs x'\|_q^\pE - \rho(\epsilon, \|\bs x - \bs x'\|_q)
\leq \Demb\big(\Amap(\bs
  x), \Amap(\bs
  x')\big) \leq (1+\epslin) \|\bs x - \bs x'\|_q^\pE + \rho(\epsilon, \|\bs x - \bs x'\|_q).
\end{equation} 
\end{definition}

In~\eqref{eq:QRIP-def}, we do not impose $\cl D_{\cl E}$ to be a function of some distance in $\bb R^m$. Actually, the role of this pre-metric when combined with $\Amap$ is to give us a computable quantity, \ie a measure of dissimilarity between $\Amap(\bs x)$ and $\Amap(\bs x')$ that approximates a power of the initial $\ell_q$-distance separating $\bs x$ and $\bs x'$. Note also that this definition is intentionally general in order to simplify
the presentation of our results. In particular, this paper
will consider only the cases in which $\pE=1$ or 2, and the
pre-metric $\Demb$ will be proportional to the $\ell_1$-distance, the squared
$\ell_2$-distance in~$\bb R^m$ or to a pre-metric derived from the $\ell_1$-distance
as will be clearer below. 

We shall also remark that the targeted QRIP,
through the requirements imposed on $\rho(\epsilon, s)$, \ie its disappearance when
$s=0$ and $\epsilon = 0$ or its decay relatively to $s^\pE$ for large
values of~$s$, is significantly stronger than
what is provided by the straightforward derivation
\eqref{eq:straightforward-quant-embed}. By raising this last relation to
the $\pE\!\th$ power, as in~\eqref{eq:QRIP-def}, we can see how the resulting additive
distortion $\rho$ contains an independent term in $\delta^\pE$ that does not vanish when both $\epsilon$ and
the distance $\|\bs x - \bs x'\|_q$ tend to zero. 
In fact, the hypothesis that $\rho(0,0)=0$ and $\lim_{s \to + \infty} s^{-\pE} \rho(\epsilon,
s) = 0$ arises from the need for our definition to allow cases where $\rho(0,s) \propto s^{p'_{\cl E}}$ for some $0<p'_{\cl E}<p_{\cl E}$ (as in Prop.~\ref{prop:distorted-main}). This situation does still provide a stronger embedding than the one derived in \eqref{eq:straightforward-quant-embed}, \ie with an additive distortion that becomes arbitrarily small for arbitrarily close $\bs x,\bs x' \in \cl K$. However, two of the three quantized embeddings proposed in this work will have $\rho(\epsilon,s) \lesssim \epsilon$ (as in Prop.~\ref{prop:main} and Prop.~\ref{prop:bi-dithered-embedding}), \ie with an additive distortion that can be made arbitrarily small with $\epsilon>0$.

\medskip
In this clarified context, the gist of this paper is to show that we can leverage the
$(\ell_p,\ell_q)$-RIP$(\cl K, \epslin)$ of certain matrix
constructions, if this holds for some $p,q\geq 1$, to prove that, with high
probability and when $m$ is large compared to the complexity of $\cl
K$, the random map $\Amap$ in~\eqref{eq:quant-embed-def}
respects a $(\Demb, \ell_q)$-QRIP$(\cl K, \epslin,\rho)$,
for suitable $\Demb$ and additive distortion $\rho$. 
From \eqref{eq:basic-dith-quant-prop}, the dither $\bs \xi$ introduced in the definition of $\Amap$ plays a crucial
role in this analysis. It enables us to reach a distortion $\rho$ that decays when the embedding dimension $m$
increases and/or when the pairwise distances of vectors in $\cl K$ decrease. 

\medskip
As an important ingredient, the main requirement on $m$ on which
our results are built depends on the dimension of the considered low-complexity set $\cl K \subset \bb
R^n$, as characterized by its Kolmogorov entropy. In the proofs, this allows us to bound the cardinality of its \emph{covering}
according to the set's \emph{radius}. In particular, given $q \geq 1$ and assuming
$\cl K$ is bounded, an
$(\ell_q,\eta)$-net $\cl K_{\ell_q, \eta} \subset \cl K$ of radius $\eta>0$ in the
$\ell_q$-metric is a finite set of vectors that \emph{covers} $\cl K$,
\ie 
$$ 
\forall \bs x \in \cl K,\quad \min\{\|\bs x - \bs u\|_{q}: \bs u \in \cl
K_{\ell_q, \eta}\} \leq \eta.
$$
Assuming that $\cl
K_{\ell_q, \eta}$ is optimal, \ie that there
is no smaller set covering $\cl K$, the Kolmogorov entropy of $\cl K$ is then defined
as 
$$
\cl H_q(\cl K, \eta) := \log |\cl K_{\ell_q, \eta}|.
$$    

While we are not aware of a general procedure to bound $\cl H_q$ for $ q\neq 2$ and any bounded set $\cl K \subset \bb R^n$, in the case $q=2$ the Sudakov minoration
provides~\cite{chandrasekaran2012convex,plan2014dimension} 
\begin{equation}
  \label{eq:Sudakov-minoration}
  \cl H_2(\cl K, \eta) \leq {\eta}^{-2} w(\cl K)^2,  
\end{equation}
where 
\begin{equation}
  \label{eq:1}
  w(\cl K) := \bb E_{\bs g} \sup\{|\scp{\bs g}{\bs u}|: \bs u \in \cl K\}, 
\end{equation}
with $\bs g \sim
\cl N^n(0,1)$, is the \emph{Gaussian mean width} of $\cl K$. This last quantity is known for many sets on which
directly estimating $\cl H_2$ is not easy 
\cite{chandrasekaran2012convex,plan2014dimension,plan2013robust,jacques2015small}. This 
includes the convex set of compressible vectors $\cl C_{s} := \sqrt s\,\bb B_{\ell_1}^n \cap \bb B^n$ with $w(\cl C_s)^2 \lesssim s \log (n/s)$, as well as the set of $n_1 \times n_2$ rank-$r$ matrices $\cl R_{r} := \{\bs U \in \bb
R^{n_1  \times n_2} \cong \bb R^n: \rank \bs U \leq r\}$ with $n = n_1n_2$ such that $w(\cl R_r \cap \bb B^{n_1 \times n_2}_{\ell_F})^2 \lesssim r (n_1 + n_2)$ where $\bb B^{n_1 \times n_2}_{\ell_F}$ is the Frobenius unit ball.  
In addition, if $\cl K \subset \bb R^n$ is a finite union of $T$ bounded sets $\cl K_i \subset \bb R^n$ with $\max_{i \in [T]} w(\cl
K_i) \leq W$ for some $W>0$, we can show (see Lemma~\ref{lem:gaussian-mean-width-union-set} in App.~\ref{sec:gmw-rank-jsparse}) that 
$w(\cl K)^2\ \lesssim\ W^2 + \log T$, which provides a bound for $\cl H(\cl
K, \eta)$ through \eqref{eq:Sudakov-minoration}. More directly, if all those sets have bounded Kolmogorov entropy, \ie if $\max_i \cl H(\cl K_i,\eta) \leq H(\eta)$ for some positive function $H$ and all $\eta > 0$, then $\cl H(\cl K, \eta) \lesssim H(\eta) + \log T$.

However, note that for some of the aforementioned sets and other \emph{structured} sets listed below, the entropy bound~\eqref{eq:Sudakov-minoration} can be significantly tightened for
small values of $\eta$, \ie essentially reducing the dependence in $\eta$ of~\eqref{eq:Sudakov-minoration} from\footnote{In fact, $\cl H_q(\cl K, \eta) \lesssim \log(1/\eta)$ for all spaces $\cl K$ having a finite ``box-counting dimension'' in the $\ell_q$-metric, \ie if $\dim \cl K := \lim\sup_{\eta \to 0^+} \cl H_q(\cl K, \eta)/\log(1/\eta)$ is bounded \cite{puy2015recipes}. As explained therein, this includes the structured sets cited in our paper as well as compact Riemannian manifolds \cite{eftekhari2015new}.} $\eta^{-2}$ to $\log(1/\eta)$. Moreover, for $q \geq 1$, there exist a few examples of structured sets $\cl K$ where $\cl H_q$ can be bounded. This is described hereafter and summarized in Tab.~\ref{tab:Kolmog-bounds}.

\begin{table}[!t]
  \centering

\scriptsize
  \begin{tabular}{p{4.9cm}|@{\ }c@{\ }|p{4cm}|p{4.5cm}}
    {\bf Set $\cl K$}&$q$&{\bf Upper-bound on $\cl H_q(\cl K, \eta)$}&{\bf Upper-bound on $w(\cl K \cap \bb B^n)^2$}\\
    \hline
&&&\\[-3mm]
{\em Structured sets:}&&&\\
-- ULS with $T$ s-dimensional\newline
\phantom{-- }subspaces of $\bb R^n$~\cite{blumensath2011sampling}&$q\geq 1$&$s\,(1+\log (T^{\frac{1}{s}}))\,\log(1+\frac{\|\cl K\|_q}{\eta})$&$(s + \log T) \lesssim s\,(1 + \log T^{\frac{1}{s}})$\newline (from \cite[Lemma 2.3]{plan2013robust})\\[2mm]
-- $s$-sparse signals in $\bb R^n$:\newline
\phantom{-- }$\Sigma_{s,n}\cap \bb B_{\ell_q}^n$&$q\geq
                                     1$&$s\log(\frac{n}{s})\log(1+\frac{1}{\eta})$&$s\log(\frac{n}{s})$\\[1mm]
-- $s$-sparse signals in a dictionary:\newline
\phantom{-- }$\bs D\Sigma_{s,d}\cap \bb B_{\ell_q}^n$, 
$\bs D \in \bb R^{n \times d}$~\cite{candes2011compressed,rauhut2008compressed}
&$q\geq 1$&$s\log(\frac{d}{s})\log(1+\frac{1}{\eta})$&$s\log(\frac{d}{s})$\\[5mm]
-- $(s,l)$-group-sparse signals:\newline
\phantom{-- }$\Sigma_{s,l,n}\cap \bb B_{\ell_q}^{nl}$~\cite{oymak2015near,ayaz2016uniform}&$q\geq 1$&$s\,(l+\log \tfrac{n}{s})\,\log(1+\frac{1}{\eta})$&$s\,(l+\log \tfrac{n}{s})$\\[1mm]
-- $n_1 \times n_2$ rank-$r$ matrices ($n_1n_2=n$):\newline
\phantom{-- }$\cl R_r \cap \bb B^{n_1 \times n_2}_{\ell_F}$~\cite{candes2011tight,chandrasekaran2012convex}&$q=2$&$r (n_1 + n_2) \log(1 + \frac{1}{\eta})$&$r (n_1 + n_2)$\\ 
-- $n_1 \times n_2$ rank-$r$ matrices in $\bb B^{n_1 \times n_2}_{\ell_F}$\newline
\phantom{-- }with at most $s$ non-zero rows \newline
\phantom{-- }(\ie $s$-joint sparse) (\cite{golbabaee2012compressed} \& App.~\ref{sec:gmw-rank-jsparse})&$q=2$&$\big( r (s + n_2) + s \log(\frac{n_1}{s}) \big) \times$\newline
\phantom{-----}$\log(1 + \frac{1}{\eta})$&$r (s + n_2) + s \log(\frac{n_1}{s})$\\ 
&&&\\[-3mm]
\hline
&&&\\[-3mm]
{\em Other sets:}&&Sudakov's minoration:&Dudley's inequality:\\
--  bounded set $\cl K \subset \bb R^n$ (\eg \cite{chandrasekaran2012convex,plan2014dimension})&$q=2$&$\eta^{-2} w(\cl K)^2$&$\int_0^{+\infty} \cl H(\cl K, \eta)^{1/2}\,\ud\eta$\\[1mm]
--  $\cl K = \cup_{i=1}^T \cl K_i \subset \bb R^n$\newline 
\phantom{-- }with $\max_i w(\cl K_i) \leq W$\newline
\phantom{-- }and $\max_i \cl H(\cl K_i,\eta) \leq H(\eta)$.
&$q=2$&$H(\eta) + \log T$\newline
(Lemma~\ref{lem:gaussian-mean-width-union-set}, App.~\ref{sec:gmw-rank-jsparse})&$W^2 + \log T$\newline
(Lemma~\ref{lem:gaussian-mean-width-union-set}, App.~\ref{sec:gmw-rank-jsparse})\\[5mm]
--  compressible signals in $\bb R^n$:\newline 
\phantom{-- }$\cl C_s \cap \bb B^n$ \cite{plan2014dimension}&$q=2$&$\eta^{-2} s \log(\frac{n}{s})$&$ s\log(\frac{n}{s})$\\[2mm]
  \end{tabular}

  \caption{Non-exhaustive list of sets with known upper-bounds on their Kolmogorov entropy and their Gaussian mean width. References are provided for bounds that are not proved in this paper. These upper-bounds are valid up to a multiplicative constant $C > 0$ that is independent of the involved dimensions. For the set of bounded rank-$r$ matrices, the $\ell_2$-norm (\ie $q=2$) is identified with the Frobenius norm (up to a vectorization).}
  \label{tab:Kolmog-bounds}
\end{table}

\begin{itemize}
\item \emph{Union of low-dimensional subspaces (ULS)} (for $q\geq 1$): Let
  us consider the ULS model where $\cl K=
  \cup_{i=1}^T \cl S_i$ is a union of $T$ $s$-dimensional subspaces
  $\cl S_i$ in $\bb B^n_{\ell_q}$ \cite{blumensath2011sampling}. This model includes, for instance, a
  single subspace of $\bb R^n$, the set $\bs D\Sigma_{s,d} \cap \bb B^n_{\ell_q}$ of sparse signals in 
some dictionary $\bs D \in \bb R^{n \times d}$ with $d \geq n$ (with
orthonormal bases and frames as
special cases) where\footnote{From the
Stirling approximation ${d \choose s} \leq (ed/s)^s$.} $\log T \leq s \log (\tfrac{n}{s})$, or even the set of $(s,l)$-group-sparse signals $\Sigma_{s,l,n} := \{\bs x \in \bb R^{ln}: |\{i: \|\bs x_{\Omega_i}\| \neq 0\}| \leq s\}$ in $\bb R^{nl}$ associated to the partition $\{\Omega_i \subset [ln]: 1\leq i \leq n, |\Omega_i| \leq l\}$ \cite{oymak2015near,ayaz2016uniform}. Under the ULS
model with $\eta \in (0,1)$ we have 
$$
\ts \cl
H_q(\cl K, \eta) \leq  \log(T) + s\log (1+ \frac{2}{\eta}) \leq s\,(1+\log(T^{1/s})) \log (1+ \frac{2}{\eta}).
$$
For instance, this gives $\cl H_q(\Sigma_{s,n}\cap \bb B_{\ell_q}^n, \eta)\ \lesssim\ s \log(\frac{n}{s}) \log(1+ \frac{2}{\eta})$. This can
be easily established by observing that an $(\ell_q,\eta)$-net of
each of such $s$-dimensional subspaces in $\bb B^n_{\ell_q}$ has
no more than $(1+ \tfrac{2}{\eta})^s$ elements (see, \eg~\cite[Lemma 4.10]{pisier1999volume}). Therefore, a
(non-optimal) covering of their union can be obtained from the union
of the nets, whose cardinality is smaller than $T\,(1+
\tfrac{2}{\eta})^s$, which provides the bound. 

Consequently, if $\cl K$ is a bounded ULS, we can always assume 
\begin{equation}
  \label{eq:Kolmogorov-bound-ULS}
  \ts \cl H_q(\cl K,\eta)\ \lesssim\ C_{\cl K}\,\log (1 + \frac{\|\cl K\|_q}{\eta}),  
\end{equation}
where $\|\cl K\|_q := \sup\{\|\bs u\|_q: \bs u \in \cl K\}$ is the $\ell_q$-diameter of $\cl K$ and
with $C_{\cl K}>0$ only depending on $\cl K \cap \bb B_{\ell_q}^n$ (\eg
$C_{\cl K} \lesssim s \log n/s$ for $\cl K = \Sigma_{s,n}\cap r\bb B_{\ell_q}^n$ for $r > 0$).

Henceforth, we say that a set is \emph{structured} if it is \emph{homogeneous} (\ie if $\lambda \cl K \subset \cl K$ for all $\lambda \geq 0$) and if
\eqref{eq:Kolmogorov-bound-ULS} is respected. Notice that all $\cl K$ that respect a ULS model are necessarily
\emph{homogeneous} and  also \emph{symmetric}, \ie $\cl K = - \cl K$. 
\item \emph{Other structured sets} (for $q=2$): In the case $q=2$, the analysis
  above extends to other structured
sets $\cl K$, as explained, \eg in
\cite{oymak2015near}. For instance, for the set of bounded rank-$r$ matrices in
$\bb R^{n_1 \times n_2} \cong \bb R^{n}$ with $n_1n_2 = n$ (see, \eg~\cite[Lemma
3.1]{candes2011tight}, \cite{moshtaghpour2016consistent}), the set of bounded rank-$r$ and jointly $s$-sparse $n_1 \times n _2$ matrices having only~$s$ non-zero rows (see App.~\ref{sec:gmw-rank-jsparse} and \cite{golbabaee2012compressed}) or even more
advanced models using group-sparsity, we still have 
$\cl H_2(\cl K,
\eta) \lesssim C_{\cl K}\,\log (1 + \frac{\|\cl K\|_2}{\eta})$ as in \eqref{eq:Kolmogorov-bound-ULS}, 
for some $C_{\cl K}>0$ depending only on $\cl
K \cap \bb B^n$. Actually, as shown in \cite{oymak2015near,chandrasekaran2012convex} (see also \cite[Def. 1]{jacques2015small}), for the bounded and structured sets listed above we have either that the constant $C_{\cl K}$ is proportional to $w(\cl K \cap \bb B^n)^2$ or that both quantities share a similar simplified upper bound. It is therefore assumed in these works that for structured sets 
  \begin{equation}
    \label{eq:struct-set-kolm-bound}
    \ts \cl H_2(\cl K, \eta)\ \lesssim\ w(\cl K \cap \bb B^n)^2 \log(1 + \tfrac{\|\cl K\|_2}{\eta}).     
  \end{equation}
This assumption can be observed in a few examples of sets provided in Table~\ref{tab:Kolmog-bounds}. 
\end{itemize}
Note that in the case of the structured sets listed in Table~\ref{tab:Kolmog-bounds}, $\cl H_q(\cl K, \eta)$ for $q \geq 1$ has the same upper-bound than the quantity $w(\cl K \cap \bb B^n)^2 \log(1 + \frac{\|\cl K\|_q}{\eta})$. This hints that in these cases $C_{\cl K}$ in \eqref{eq:Kolmogorov-bound-ULS} and $w(\cl K \cap \bb B^n)^2$ also share similar upper-bounds.


\medskip
Our first main result analyzes the following RIP inheritance.

\begin{proposition}[QRIP inherited from the $(\ell_1, \ell_q)$-RIP] 
\label{prop:main}
Let us fix $\epslin, \epsilon \in (0, 1)$ and $q \geq 1$. Assume that $\bs \Phi \in \bb R^{m \times n}$
respects the $(\ell_1, \ell_q)$-RIP$(\cl K - \cl K, \epslin)$. If  
\begin{equation}
  \label{eq:cond-on-m-main-prop}
  \ts m \gtrsim \epsilon^{-2} \cl H_q(\cl K, \delta \epsilon^2),    
\end{equation}
then, for some $C,c>0$, the quantized random embedding $\Amapd$ satisfies the
$(\Demb,\ell_q)$-QRIP$(\cl K, \epslin, \rho)$ with $\rho \lesssim
\delta \epsilon$, $\Demb = \cl D_{\ell_1}$
(\ie $\pE = 1$) and with probability exceeding $1- C\exp(-c
m \epsilon^2)$. In other words, under these conditions and with the
same probability,
\begin{equation}
  \label{eq:QRIP-first-result}
(1-\epslin) \|\bs x - \bs x'\|_q - c \delta \epsilon 
\leq \tinv{m} \|\Amap(\bs x) - \Amap(\bs x')\|_1 \leq (1+\epslin) \|\bs
x - \bs x'\|_q + c\delta \epsilon, 
\end{equation} 
for all $\bs x,\bs x' \in \cl K$ and some $c>0$.

\end{proposition}

The full proof of this proposition is postponed to Sec.~\ref{proof-of-main-prop}. However,
its sketch is intuitively simple and worth outlining as follows. From~\eqref{eq:basic-dith-quant-prop}, we first notice that 
$\bb E_{\bs \xi}(\tinv{m} \|\Amap(\bs x) - \Amap(\bs x')\|_1) = \|\bs \Phi(\bs
x - \bs x')\|_1$. We understand then that, for $\bs x, \bs x'$ fixed and provided we collect enough observations for $\tinv{m} \|\Amap(\bs x) - \Amap(\bs x')\|_1$ to approach $\bb E_{\bs \xi}(\tinv{m} \|\Amap(\bs x) - \Amap(\bs x')\|_1)$ with a controlled error, it should be possible to exploit the $(\ell_1, \ell_q)$-RIP$(\cl K - \cl K, \epslin)$ of $\bs \Phi$ to relate $\|\bs \Phi(\bs
x - \bs x')\|_1$, and thus $\tinv{m} \|\Amap(\bs x) - \Amap(\bs x')\|_1$, to $\|\bs x - \bs x'\|_q$ up to a multiplicative distortion involved by this linear embedding. 
In fact, using measure concentration on the randomness induced by the dithering, with
probability exceeding $1- C\exp(-c'
m \epsilon^2)$, the (pseudo) distance $\cl D(\bs a + \bs \xi, \bs a' + \bs
\xi) := \Demb(\cl Q(\bs a + \bs \xi), \cl Q(\bs a' + \bs \xi))$
concentrates around $\cl D_{\ell_1}(\bs a, \bs a')$  with an
additive distortion in $\delta\epsilon$ for a fixed pair
of vectors $\bs a = \bs \Phi \bs x$ and $\bs a' = \bs \Phi \bs x' \in \bb R^m$. 

It remains then to extend this result to vectors taken in the full set $\bs \Phi \cl K \subset \bb
R^m$ and then to use the $(\ell_1,\ell_q)$-RIP$(\cl K-\cl K, \epslin)$ to connect $\cl
D_{\ell_1}(\bs a, \bs a')$ to
$\|\bs x - \bs x'\|_q$, up to a multiplicative distortion~$1\pm \epslin$.
If $\Amapd$ was linear (\ie if $\cl Q = \Idf$) this would be done by covering the set $\cl K$, here according
to the $\ell_q$-metric, with a
sufficiently dense but finite set of vectors (\ie a $\eta$-net with radius
$\eta>0$). Knowing that
the set cardinality is bounded by $\exp(\cl H_q(\cl K, \eta))$, a
union bound provides that the concentration of the \rv $\Demb(\Amapd(\bs x),
\Amapd(\bs x'))$ can hold simultaneously for all vectors
of the $\eta$-net if~\eqref{eq:cond-on-m-main-prop} is
respected. This can be extended to $\cl K$ by continuity
of the linear map (see \eg~\cite{baraniuk2008simple} for an application of this
method to proving the RIP of sub-Gaussian matrices). However, the
argument breaks here as $\Demb(\Amapd(\cdot),
\Amapd(\cdot))$ is actually discontinuous due to the presence of $\cl Q$.  

This issue can be
fortunately overcome by \emph{softening} the impact of the quantizer and the distance between the mapped vectors. We follow for that a method initially developed in~\cite{plan2014dimension} for 1-bit quantization (with a sign operator) and extended later
in~\cite{jacques2015small} to uniform scalar quantization (see Sec.~\ref{proof-of-main-prop}). Roughly speaking, this is done by first observing that for any scalar $a, a' \in \bb R$, the distance $\delta^{-1} d(a,a') := \delta^{-1} |\cl Q(a) - \cl Q(a')|$, which appears in the decomposition of $\cl D_{\ell_1}$ over the $m$ components of the mapping $\Amapd$, amounts to counting the number of \emph{quantizer thresholds} in $\delta \bb Z$ that separate the non-quantized values $a$ and $a'$.  This counting process, which is intrinsically discontinuous with respect to $a$ and $a'$, is then softened by introducing a parameter $t \in \bb R$ whose value determines forbidden (or relaxed) intervals $\delta \bb Z + [-|t|, |t|]$ if $t > 0$ (resp. $t < 0$) of size $2|t|$ and centered on the quantizer thresholds in $\delta \bb Z$. Compared to the definition of $\delta^{-1} d$, which corresponds to $t=0$ and no softening, for $t >0$ we will {\em not count} a threshold of $\delta \bb Z$ between $a$ and $a'$ when it falls in the \emph{forbidden} interval (with $t > 0$), whereas we will {\em count} a threshold that is {\em not} between $a$ and $a'$ if it falls in the \emph{relaxed} interval (with $t < 0$). Interestingly, the distance $d^t(a,a')$ resulting from such a softened counting process gives then access to a form of continuity that $d=d^0$ is deprived of, \ie one can show that for all $r,r' \in [-r_0,r_0]$ and some $r_0 > 0$, $d^{t+r_0}(a,a') \leq d^{t}(a+r,a'+r') \leq d^{t-r_0}(a,a')$ (this is carefully detailed in Lemma~\ref{lem:cont-smald-t}). 
This suffices to connect the machinery developed to verify the RIP of linear mapping with that required for quantized random embeddings. To do so, we first extend the continuity of $d^t$ to that of $\cl D^t$ to $\bb R^m$ (with $\cl D = \cl D^0$), which gives a softening of $\cl D_{\ell_1}(\Amapd(\cdot),
\Amapd(\cdot))$. Then, we construct a $\eta'$-net of $\bb R^m$ in the $\ell_1$-metric from the image by $\bs \Phi$ of an $\eta$-net of $\cl K$ in the $\ell_q$-metric with the same cardinality, enjoying for this of the $(\ell_1,\ell_q)$-RIP of $\bs \Phi$ to connect $\eta'$ to $\eta$. Finally,~\eqref{eq:QRIP-first-result} can be proved by verifying it for all vectors in the neighborhoods that form the $\eta'$-net of $\bb R^m$. This is achieved by a perturbation analysis of $\cl D^t$ around the centers of these neighborhoods for which~\eqref{eq:QRIP-first-result} can be shown by union bound, using the continuity that $d^t$ confers on $\cl D^t$ in the said neighborhoods. The final result is obtained by setting $t=0$ at the very end of the developments.  
\medskip

As a second result we are able to show, up to a stronger distortion 
$\rho(\epsilon, s) \lesssim \delta s + \delta^2 \epsilon$ that only
vanishes when \emph{both} $\epsilon$ and $s$ tend to zero, that the $(\ell_2,\ell_2)$-RIP
implies a $(\Demb,\ell_2)$-QRIP, with $\Demb =
\cl D_{\ell_2}$. 

\begin{proposition}[QRIP inherited from the
  $(\ell_2,\ell_2)$-RIP] 
\label{prop:distorted-main}
Let us fix $\epslin, \epsilon \in (0,1)$. Assume that $\bs \Phi \in \bb R^{m \times n}$
  respects the $(\ell_2,\ell_2)$-RIP$(\cl K - \cl K, \epslin)$. If  
  \begin{equation}
    \label{eq:cond-on-m-distorted-main-prop}
    m \gtrsim \epsilon^{-2} \cl H_2(\cl K, \delta \epsilon^{3/2}),    
  \end{equation}
then, for some
$C,c>0$ and with probability exceeding $1- C\exp(-c m \epsilon^2)$, the quantized random embedding $\Amapd$ satisfies the $(\Demb,\ell_2)$-QRIP$(\cl K, \epslin, \rho)$, with $\Demb = \cl D_{\ell_2}$ (\ie $\pE = 2$) and $\rho(\epsilon, s)
\lesssim \delta s + \delta^2 \epsilon$. 
In other words, under these conditions and with the
same probability,
\begin{equation}
  \label{eq:QRIP-2nd-result}
\big| \tinv{m} \|\Amap(\bs x) - \Amap(\bs x')\|^2_2 - \|\bs x - \bs
x'\|^2_2 \,\big|\ \lesssim\ \epslin\,\|\bs x - \bs x'\|^2_2 + \delta \|\bs x - \bs x'\|_2 +
\delta^2 \epsilon, 
\end{equation} 
for all $\bs x,\bs x' \in \cl K$.
\end{proposition}

The proof architecture of Prop.~\ref{prop:distorted-main},
developed in Sec.~\ref{sec:proof-distorted-embedding}, is similar to the
one of Prop.~\ref{prop:main} except that special care must be given to
the quadratic nature of $\Demb = \cl D_{\ell_2}$ and of the softened
version of $\cl D_{\ell_2}(\Amapd(\cdot),\Amapd(\cdot))$. 
 
Notice that the cause of such a strongly distorted embedding can be traced back to the expectation of
$\bb E_{\xi} |\cl Q(a + \xi) - \cl Q(a' + \xi)|^2$ for $a,a' \in \bb R$,
which determines the expectation of $\tinv{m} \|\Amap(\bs x) - \Amap(\bs
x')\|^2_2$ by separability of $\cl D_{\ell_2}$. Indeed, if $|a-a'|<\delta$, the random variable $X :=
\tinv{\delta}|\cl Q(a + \xi) - \cl Q(a' + \xi)|\in \{0,1\}$ is
binary, \ie $\bb E X^q = \bb E X$ for all $q > 0$. Thus $\bb E_{\xi} |\cl Q(a + \xi) - \cl Q(a' + \xi)|^2 = \delta \bb
E_{\xi} |\cl Q(a + \xi) - \cl Q(a' + \xi)| = \delta |a - a'|$, which
clearly deviates from the quadratic dependence in $|a-a'|$ reached by the lower bound $\bb E_{\xi} |\cl
Q(a + \xi) - \cl Q(a' + \xi)|^2 \geq (\bb E_{\xi} |\cl
Q(a + \xi) - \cl Q(a' + \xi)|)^2 = |a - a'|^2$ (using Jensen and~\eqref{eq:basic-dith-quant-prop}). In
this view, $\Demb(\Amap(\cdot), \Amap(\cdot))$ cannot escape a
strong deviation with respect to the quadratic metric $\cl D_{\ell_2}(\cdot, \cdot)$ in the asymptotic regime
(\ie for large values of $m$) when pairwise signal distances are small
compared to~$\delta$. 

This change of regime between small- and large-distance embeddings also
meets similar behaviors described in a version of the Johnson-Lindenstrauss lemma
for the quantized embedding of finite sets by the same non-linear map in~\eqref{eq:quant-embed-def}~\cite[Sec. 5]{jacques2015quantized} or in
more complex non-linear random embeddings as observed in
\cite{boufounos2015representation} (see Sec.~\ref{sec:comp-with-other}).

\paragraph{Bi-dithered quantized maps:}
As a last result of this paper, we provide a novel and simple way
to extend the dithering procedure of~\eqref{eq:quant-embed-def} in
order to improve the distortion properties of the quantized embedding
of Prop.~\ref{prop:distorted-main} and yet preserve the inheritance of the $(\ell_2,\ell_2)$-RIP of the linear map $\bs \Phi$. This new scheme relies on a \emph{bi-dithered quantized map}
which, up to an implicit doubling of the number of measurements, is able to remove the strong distortion of the
previous $(\cl D_{\ell_2},\ell_2)$-quantized embedding, \ie reducing it from
$\rho \lesssim \delta s + \delta^2\epsilon$  to
$\rho \lesssim \delta^2\epsilon$. 

The best way to introduce it is to first generalize~\eqref{eq:quant-embed-def} to a matrix
map over $\bb R^{n \times 2}$, \ie we now define $\Amapd: \bs X \in \bb R^{n\times 2} \to \bb Z_\delta^{m
  \times 2}$ such that
\begin{equation}
  \label{eq:bi-dithered-mapping}
  \Amapd(\bs X) = \cl Q( \bs \Phi \bs X + \bs \Xi),  
\end{equation}
where $\cl Q$ 
is applied entrywise on its matrix argument, $\bs \Phi \in \bb
R^{m\times n}$ is the linear part of the map $\Amapd$ and the dither $\bs \Xi \sim \cl U^{m \times
  2}([0,\delta])$ is now composed of two independent columns. 

To characterize this embedding we proceed by introducing  the operation $(\cdot)^\circ:\bs B \in \bb
R^{m\times d} \to \bs B^\circ \in \bb R^m$ such that 
$(\bs B^\circ)_i = \Pi_{j=1}^d B_{ij}$ together with the pre-metric $\|\bs B\|_{1,\circ} := \|\bs
B^\circ\|_{1}$, which is a norm only for $d=1$. It is easy to see
from the independence of the columns of $\bs \Xi$ that 
$$
\bb E_{\bs \Xi} \|\Amapd(\bs X) - \Amapd(\bs X')\|_{1,\circ} = \|\bs \Phi(\bs X
- \bs X')\|_{1,\circ}.
$$    
Indeed, from the independence of the entries of $\bs \Xi$ and
  from~\eqref{eq:basic-dith-quant-prop} we find
  \begin{align*}
\ts \bb E_{\bs \Xi} \|\cl Q( \bs \Phi \bs X + \bs \Xi) - \cl Q( \bs \Phi
\bs X' + \bs \Xi)\|_{1,\circ}&\ts = \sum_i \Pi_{j=1}^d \bb E_{\bs \Xi}
                               |\cl Q((\bs \Phi \bs X)_{ij} +
                               \Xi_{ij}) - \cl Q((\bs \Phi \bs
                               X')_{ij} + \Xi_{ij})|\\
&\ts= \sum_i \Pi_{j=1}^d  |\bs \Phi (\bs X - \bs X')_{ij}| = \| \bs
  \Phi (\bs X - \bs X')\|_{1,\circ}.        
  \end{align*}
In particular, using the compact notation $\bar{\bs u} := \bs u \bs
1_2^\transp \in \bb R^{n \times 2}$ for any vector $\bs u \in \bb
R^n$, and defining the quantized map 
$$
\bar{\Amapd}:\ \bs x \in \bb R^n\ \mapsto\ \bar{\Amapd}(\bs x) := \Amapd(\bar{\bs x}) \in
\bb Z_\delta^{m \times 2},
$$
we have  
$$
\bb E_{\bs \Xi} \|\bar{\Amapd}(\bs x) - \bar{\Amapd}(\bs x')\|_{1,\circ} = \|\bs \Phi(\bs x
- \bs x')\|^2.
$$  
Therefore, up to the identification $\bb
Z_\delta^{m \times 2} \cong \bb
Z_\delta^{2m}$, the map $\bar{\Amapd}$ amounts to doubling the number of
measurements $m$, with the use of two random dithers (\ie the columns of $\bs \Xi$) per action of a
single $\bs \Phi$ on $\bs x$, compared to the single dither of the
initial map~\eqref{eq:quant-embed-def}. As described below this is, however, highly beneficial for reducing the
distortion of the quantized embedding of $(\cl K,
\ell_2)$ in $(\bb Z^{m \times 2}_\delta, \cl D_{\circ})$ with $\cl
D_{\circ}(\cdot) := \tinv{m}\|\cdot\|_{1,\circ}$, with this improvement being achieved with
high probability by the map
$\bar{\Amapd}$ once $\bs \Phi$ has the $(\ell_2,\ell_2)$-RIP. 
\begin{proposition}[$(\ell_2,\ell_2)$-RIP involves bi-dithered quantized embedding]
\label{prop:bi-dithered-embedding}
Let us fix $\epslin,\epsilon \in (0,1)$. If $\bs \Phi \in \bb R^{m \times n}$
  respects the $(\ell_2,\ell_2)$-RIP$(\cl K - \cl K, \epslin)$ and  
  \begin{equation}
    \label{eq:cond-on-m-bidither-prop}
\ts    m \gtrsim \epsilon^{-2} \cl H_2(\cl K, \delta \epsilon^{2}),    
  \end{equation}
then, with probability exceeding $1-C \exp(-c m\epsilon^2)$ for some
$C,c>0$, the quantized random embedding $\bar{\Amapd}$ satisfies the
$(\Demb, \ell_2)$-QRIP$(\cl K, \epslin + c\epsilon, \rho)$ with $\Demb
= \cl D_{\circ}$ (\ie $\pE = 2$) and $\rho(\epsilon,s) \lesssim
\delta^2\epsilon$. In other words, under the same conditions and with
the same probability, 
\begin{equation}
  \label{eq:bi-dithered-embedding}
  (1-\epslin - c\epsilon) \|\bs x - \bs x'\|^2 - c' \epsilon
  \delta^2 \leq \tinv{m}\| \bar{\Amapd} (\bs x) - \bar{\Amapd} (\bs
  x')\|_{1,\circ} \leq (1+\epslin + c\epsilon) \|\bs x - \bs x'\|^2 + c' \epsilon
  \delta^2,
\end{equation}
for all $\bs x,\bs x' \in \cl K$ and some $c,c'>0$.
\end{proposition}

The proof of this proposition, postponed to
Sec. \ref{sec:bi-dith-quant}, actually starts by showing the existence
of a QRIP with multiplicative distortion $\epslin$ and additive
distortion $\rho\lesssim \delta^2 \epsilon + \delta\epsilon s$ for $\epsilon \in (0,1)$. However, $\delta\epsilon s \lesssim \epsilon s^2 + \delta^2\epsilon$ and the first term of this bound can be gathered with the multiplicative distortion term in the QRIP, which explains~\eqref{eq:bi-dithered-embedding}. Fixing $\epsilon = \epslin$ and rescaling $\epsilon$ shows then that we can inherit from the $(\ell_2,\ell_2)$-RIP$(\cl K - \cl K, \epsilon)$ of $\bs \Phi$ the $(\Demb, \ell_2)$-QRIP$(\cl K, \epsilon, c \delta^2 \epsilon)$ of the bi-dithered mapping $\Amap$ for some $c>0$.

\section{Portability, Discussion and Perspectives}
\label{sec:discussion}
\begin{table}[!t]
  \centering
\scriptsize
  \begin{tabular}{@{}p{3.7cm}|@{}p{1cm}@{}|@{\,}p{1.9cm}@{\,}|@{\ }p{2.2cm}@{\,}|@{\ }p{1.8cm}@{\ }|@{\ }p{1.6cm}@{\ }|@{\ }p{1cm}@{\ }|@{\ }p{.5cm}}
    {\bf Name}&{\bf ~Ref.}&\centering{\bf Linear embedding}&{\bf
                                                            Embeddable
                                                            set $\cl
                                                            K$}&{\bf
                                                                 Encoding
                                                                 complexity}
&{\bf QRIP}&$\rho \lesssim$&\scalebox{0.8}{\bf Prop.}\newline\#\\
    \hline
    Random sub-Gaussian\newline ensembles&\cc\cite{mendelson2008uniform,klartag2005empirical}&\centering$(\ell_2,\ell_2)$&SGMW
                                                                                  in
                                                                                  $\bb
                                                                                  S^{n-1}$&$O(mn)$&$(\cl
                                                                                                    D_{\ell_2},\ell_2)^*$
           \newline or $(\cl D_{\circ},\ell_2)^*$&$\delta s + \delta^2
                                                   \epsilon$\newline $\delta^2
                                                   \epsilon$&\ref{prop:distorted-main},
    \newline \ref{prop:bi-dithered-embedding}\\
&&&&&&\\[-3mm]
\hline
&&&&&&\\[-3mm]
Random Gaussian\newline \vspace{-3mm}ensembles &\cc\cite{schechtman2006two,plan2014dimension}&\centering$(\ell_1,\ell_2)$&SGMW
                                                                             in
                                                                             $\bb
                                                                             S^{n-1}$&$O(mn)$&$(\cl
                                                                                               D_{\ell_1},
                                                                                               \ell_2)^*$&$\delta
                                                                                                           \epsilon$&\ref{prop:main}\vspace{1mm}\\ 
\cline{2-8}
&&&&&\multicolumn{2}{c}{}\\[-3mm]
    &\cc\cite{jacques2011dequantizing}&\centering$(\ell_p,\ell_2)$,\ $p
                                    \geq 1$&$\Sigma^{\rm
                                             any}_{s,n}$&$O(mn)$&\multicolumn{3}{@{\,}p{4cm}@{}}{Potential
                                                                  Multi-dithered\newline
                                                                  extension
    of Prop.~\ref{prop:bi-dithered-embedding}}\\[1mm]
&&&&&&\\[-3mm]
\hline
&&&&&&\\[-3mm]
    Random Orthonormal Basis
    ensembles&\cc\cite{rudelson2008sparse,foucart2013mathematical}&\centering$(\ell_2,\ell_2)$&$\Sigma^{\rm
                                                                                             lc}_{s,n}$&$O(n
                                                                                                         \log
                                                                                                         n)$&$(\cl
                                                                                                    D_{\ell_2},\ell_2)$
           \newline or $(\cl D_{\circ},\ell_2)$&$\delta s + \delta^2
                                                   \epsilon$\newline $\delta^2
                                                   \epsilon$&\ref{prop:distorted-main},
                                                                                                               \newline \ref{prop:bi-dithered-embedding}\\
&&&&&&\\[-3mm]
\hline
&&&&&&\\[-3mm]
Random
    convolutions&\cc\cite{rauhut2012restricted,romberg2009compressive}&\centering$(\ell_2,\ell_2)$&$\Sigma_{s,n}$&$O(n
                                                                                                                \log
                                                                                                                n)$&$(\cl
                                                                                                    D_{\ell_2},\ell_2)$
           \newline or $(\cl D_{\circ},\ell_2)$&$\delta\epsilon$&\ref{prop:distorted-main},
                                                                                                                      \newline \ref{prop:bi-dithered-embedding}\\
    \hline
&&&&&&\\[-3mm]
    Spread-spectrum&\cc\cite{puy2012universal}&\centering$(\ell_2,\ell_2)$&$\Sigma^{\rm
                                                                  any}_{s,n}$
                                                                  \whp&$O(n
                                                                        \log
                                                                        n)$&$(\cl
                                                                                                    D_{\ell_2},\ell_2)$
           \newline or $(\cl D_{\circ},\ell_2)$&$\delta s + \delta^2
                                                   \epsilon$\newline $\delta^2
                                                   \epsilon$&\ref{prop:distorted-main},
                                                                              \newline \ref{prop:bi-dithered-embedding}\\
&&&&&&\\[-3mm]
\hline
&&&&&&\\[-3mm]
 Multiresolution \scalebox{0.8}{$(\ell_2,\ell_2)$}-RIP matrix with
    random sign flipping&\cc\cite{oymak2015isometric}&\centering$(\ell_2,\ell_2)$&SGMW in
                                                                $\bb
                                                                S^{n-1}$&$O(nm)$\newline
or $O(n
                                                                        \log
                                                                        n)$&$(\cl
                                                                                                    D_{\ell_2},\ell_2)^*$
           \newline or $(\cl D_{\circ},\ell_2)^*$&$\delta s + \delta^2
                                                   \epsilon$\newline $\delta^2
                                                   \epsilon$&\ref{prop:distorted-main},
                                                                                     \newline \ref{prop:bi-dithered-embedding}\\
&&&&&&\\[-3mm]
\hline
&&&&&&\\[-3mm]
    Unbalanced expanders&\cc\cite{berinde2008combining}&\centering
                                                      $(\ell_p,
                                                      \ell_p)$,
                                                      with\newline
                                                      \scalebox{0.7}{$
                                                      1 \leq p \leq 1 + O(\frac{1}{\log
                          n})$}&$\Sigma_{s,n}$&$O(n
                                                  \log \tfrac{n}{s})$&$(\cl
                                                                                                    D_{\ell_1},\ell_1)$&$\delta
                  \epsilon$&\ref{prop:main}\\
&&&&&&\\[-3mm]
\hline
&&&&&&\\[-3mm]
Rank-$1$ random\newline 
sub-Gaussian measurements
of rank$-r$ matrices&\cc\cite{puy2015recipes,cai2015rop}&\centering $(\ell_1,\ell_F)$&$\cl R_{r}$&$O(mn)$&$(\cl D_{\ell_1},\ell_F)$&$\delta
                  \epsilon$&\ref{prop:main} \& Cor.~\ref{cor:qrip-rank-one}\\
  \end{tabular}
  \caption{Non-exhaustive list of quantized embeddings inherited from
    RIP matrix constructions. The name of the construction and its
    references, the linear embedding type, the sets that are
    embeddable by this linear embedding, the linear encoding
    complexity, the type of QRIP inherited by this linear embedding
    and the corresponding propositions in this paper are provided. In this table, SGMW means
  sets with Small Gaussian Mean Width (before $n$), $\Sigma^{\rm
    any}_{s,n}$ stands for $s$-sparse in any basis (\ie \emph{universal}
  sensing), and $\Sigma^{\rm
    lc}_{s,n}$ means $s$-sparse in a basis with low mutual coherence
  with respect to the basis used for sensing (see cited
  references). $(^*)$: Only if $\cl K$ is
  homogeneous.}
  \label{tab:Non-exhausitve-RIP-matrix-construction}
\end{table}

The main interest of the three propositions introduced in
Sec.~\ref{Main-results} is to connect the existence of quantized
embeddings of low-complexity vector sets to the one of linear embeddings
of the same sets, as derived from specific instances of the RIP. We
find it useful to emphasize in this section the portability of these
results, linking them to efficient matrix constructions satisfying the RIP (\eg with fast matrix-vector multiplication), as well
as highlighting both their connection with existing works and a few
related open
problems. 

\subsection{Portability of the results}
\label{sec:portability-results}

We provide below a brief summary of interesting RIP constructions
that prove themselves useful for at least one of our three main
propositions. The connection between these constructions, their
characteristics and the links with our main results is also summarized in
Table~\ref{tab:Non-exhausitve-RIP-matrix-construction}. Moreover, we
stress below important aspects of quantized random embeddings, such as the achievability
of fast 1-bit embeddings. 

\paragraph{(a) Random sub-Gaussian ensembles and $(\ell_2,\ell_2)$-RIP inheritance:} The first constructions of
linear maps known to respect the RIP with high probability, which
date back to the randomized maps used to prove the JL Lemma
\cite{dasgupta1999elementary,achlioptas2003database}, were Random
Sub-Gaussian Ensembles (RSGE). These are associated to a linear map $\bs \Phi
\in \bb R^{m\times n}$ with $\Phi_{ij} \sim_{\iid} X$, with $X$ a sub-Gaussian random variable with unit variance and zero expectation, such as $X \sim \cl N(0,1)$ or a Bernoulli random
variable $X \sim \cl B(\{\pm 1\})$ with $\bb P(X=\pm 1)=1/2$.   

It was rapidly established in the CS literature that, provided $m
\gtrsim \epslin^{-2}\,s \log n/s$, such a matrix $\bs \Phi$ respects the $(\ell_2,
\ell_2)$-RIP$(\bs \Psi\Sigma_{s,n}, \epslin)$ with high probability
and for any orthonormal basis (ONB) $\bs \Psi \in \bb R^{n \times n}$,
\ie RSGE are \emph{universal} in that sense~\cite{candes2005decoding}.
Hence, RSGE is of course a first class of matrix constructions that can
be used for inducing dithered or bi-dithered quantized embeddings of $\bs \Psi\Sigma_{s,n}$, as
stated in Prop.~\ref{prop:distorted-main} and
Prop.~\ref{prop:bi-dithered-embedding}, respectively. 

However, it is also observed
in~\cite{mendelson2008uniform,klartag2005empirical} that for any set
$\cl K \subset \bb S^{n-1}$, we have, with probability exceeding $1-\exp(-c \epslin^2m)$,
\begin{equation}
  \label{eq:RIP2-embed-SLGW}
\sup_{\bs u \in \cl K} |\tinv{m}\|\bs\Phi \bs u\| - 1| \leq \epslin,  
\end{equation}
provided $m \gtrsim \epslin^{-2} w(\cl K)^2$, with $w(\cl K)$ the
Gaussian mean width defined in Sec.~\ref{Main-results}. This provides the following easy corollary. 
  \begin{corollary}
\label{cor:rip22-RGE}
    Let $\cl K$ be a \emph{homogeneous} subset of  $\bb R^n$, \ie $\lambda \cl K
\subset \cl K$ for all $\lambda \geq 0$, and define $\cl K^* := {\cl K \cap \bb
B^n}$. Given $\epslin \in (0,1)$, if 
\begin{equation}
  \label{eq:RIP22-cond-m}
  m \gtrsim \epslin^{-2}
  w(\cl K^*)^2,
\end{equation}
then, with probability exceeding $1-\exp(-c \epslin^2m)$, a RSGE matrix $\bs \Phi \in \bb R^{m\times n}$ respects the $(\ell_2,
\ell_2)$-RIP$(\cl K - \cl K, \epslin)$. 
  \end{corollary}
  \begin{proof}
    For a homogeneous set $\cl
K$, $\bs u /\|\bs u\|
\in {\cl K \cap \bb S^{n-1}}$ for all $\bs u \in \cl K \setminus \{\bs 0\}$. Therefore, if $m \gtrsim \epslin^{-2}
w(\cl K^*)^2 \geq \epslin^{-2}
w({\cl K \cap \bb S^{n-1}})^2$ with $\cl K^* := {\cl K \cap \bb
B^n} \supset \cl K \cap \bb
S^{n-1}$, then~\eqref{eq:RIP2-embed-SLGW} holds on $\bs u/\|\bs u\|
\in {\cl K \cap \bb S^{n-1}}$ for all $\bs u \in \cl K$,
which shows that $\bs \Phi$ respects, \whp, the $(\ell_2,
\ell_2)$-RIP$(\cl K, \epslin)$ defined in~\eqref{eq:RIPpq}. Applying the same argument on the
homogeneous set $\cl K
- \cl K \subset \bb R^n$, and observing that ${(\cl K - \cl K)\cap \bb B^n} \subset
\cl K^* - \cl K^*$ with $w(\cl K^* - \cl K^*) \leq 2 w(\cl K^*)$, this
shows that $\bs \Phi$ also respects \whp the $(\ell_2,
\ell_2)$-RIP$(\cl K-\cl K, \epslin)$ if $m \gtrsim \epslin^{-2}
w(\cl K^*)^2$.    
  \end{proof}

The portability of Prop.~\ref{prop:distorted-main} and
Prop.~\ref{prop:bi-dithered-embedding} is therefore guaranteed for bounded vectors taken in homogeneous sets, which include
for instance the ULS model described in Sec.~\ref{Main-results} or
other structured sets such as the set of bounded low-rank
matrices in $\bb R^{n_1 \times n_2}$ with $n=n_1n_2$ (when identified with
$\bb R^n$). 

In particular, when $\cl K$ is homogeneous, it is easy to determine when the quantized mapping $\Amap$ defined in \eqref{eq:quant-embed-def} can satisfy a $(\cl D_{\cl E}, \ell_2)$-QRIP on $\cl K^* = {\cl K \cap \bb B^n}$ inherited from the $(\ell_2,\ell_2)$-RIP$(\cl K, \epslin)$ of a RSGE matrix $\bs \Phi$. Using Sudakov's minoration in \eqref{eq:Sudakov-minoration}, the requirements of Cor.~\ref{cor:rip22-RGE}, Prop.~\ref{prop:distorted-main} and Prop.~\ref{prop:bi-dithered-embedding} are indeed satisfied on $\cl K^*$ if we have 
\begin{equation}
  \label{eq:req-props2-3-on-homog-sets}
\ts m \gtrsim  \max(\frac{1}{\epslin^{2}}, \frac{1}{\epsilon^2\eta(\epsilon,\delta)^{2}})\, w(\cl K^*)^2,   
\end{equation}
while, from~\eqref{eq:struct-set-kolm-bound}, for 
structured sets these requirements are more easily verified if 
\begin{equation}
  \label{eq:req-props2-3-on-struct-sets}
\ts m \gtrsim \max\big(\,\frac{1}{\epslin^{2}},\, \log(1+\frac{1}{\eta(\epsilon,\delta)})\big)\, w(\cl K^*)^2,   
\end{equation}
with $\eta(\epsilon,\delta) = \delta \epsilon^{3/2}$ for Prop.~\ref{prop:distorted-main} and
$\eta(\epsilon,\delta) =\delta \epsilon^{2}$ for
Prop.~\ref{prop:bi-dithered-embedding}. Consequently, we see that \eqref{eq:req-props2-3-on-homog-sets} and \eqref{eq:req-props2-3-on-struct-sets} are sufficient to ensure the verification of a QRIP \whp based on a RSGE matrix $\bs \Phi$. Moreover, setting $\epsilon = \epslin$ and taking the smallest $m$ for the conditions above to hold, we see that, \whp and up to log factors, $\epsilon = \epslin = O( m^{-1/q} w(\cl K^*))$ with $q=2$ for structured sets, while $q=5$ and $q=6$ for general sets in the case of Prop.~\ref{prop:distorted-main} and Prop.~\ref{prop:bi-dithered-embedding}, respectively. 
\medskip

\paragraph{(b) Random Gaussian ensembles and $(\cl D_{\ell_1},\ell_2)$-QRIP:} A similar result to Prop.~\ref{prop:main} in the case of $q=2$, \ie
for $(\cl D_{\ell_1},\ell_2)$-quantized embeddings, was established earlier on for
dithered quantized maps specifically built on RSGE matrices~\cite{jacques2015small}. Using the notations of the present work, it was shown that the resulting map $\Amapd$ satisfies \whp the $(\cl D_{\ell_1},
\ell_2)$-QRIP$(\cl K, \epsilon, \delta^2\epsilon + \rho_{\rm sg})$ (\ie with $\epsilon = \epslin$) provided $m \gtrsim
\delta^{-2}\epsilon^{-5} w(\cl K)^2$ and with an additional distortion
$\rho_{\rm sg}>0$ for a sub-Gaussian but non-Gaussian
random map $\bs \Phi$ and whose value is small, but non-zero, for ``not too sparse''
vectors in $\cl K$, \ie those with small $\ell_\infty/\ell_2$ norm ratio. 

In the case where $\bs \Phi \sim \cl N^{m\times n}(0,1)$ is a random Gaussian matrix\footnote{Also known as random Gaussian ensembles (RGE).}, we can recover the same result from Prop.~\ref{prop:main} and from the following fact.
As an extension to \eqref{eq:RIP2-embed-SLGW}, it has been proved in
\cite{schechtman2006two} (see also~\cite{plan2014dimension}) that for 
$\cl K \subset \bb S^{n-1}$ we have, with probability exceeding $1-\exp(-c\epsilon^2m)$,
\begin{equation}
  \label{eq:RIP12-embed-SLGW}
\sup_{\bs u \in \cl K} |\tinv{m} \gone\,\|\bs\Phi \bs u\|_1 - 1| \leq \epsilon  
\end{equation}
provided $m \gtrsim \epsilon^{-2} w(\cl K)^2$, \ie as required to reach a $(\ell_2,\ell_2)$-linear embedding. 

Therefore,
applying on~\eqref{eq:RIP12-embed-SLGW} the same 
observations made on~\eqref{eq:RIP2-embed-SLGW}, this shows that for homogeneous sets $\cl K$, if 
\begin{equation}
  \label{eq:cond-m-RIP12}
  \ts m \gtrsim \epsilon^{-2}
w(\cl K^*)^2,
\end{equation}
$\bs \Phi$ respects \whp the $(\ell_1,
\ell_2)$-RIP$(\cl K - \cl K, \epsilon)$ with $\mu_\Phi = \gone$ in~\eqref{eq:RIPpq}. 

Consequently, as realized in \cite{jacques2015small}, we can determine the necessary conditions to guarantee that the quantized mapping $\Amap$
in \eqref{eq:quant-embed-def} respects \whp the $(\cl D_{\ell_1},
\ell_2)$-QRIP$(\cl K^*, \epslin = \epsilon, c\delta\epsilon)$ if $\bs \Phi$ is a random Gaussian matrix, \ie the conditions for 
satisfying both \eqref{eq:cond-m-RIP12} and \eqref{eq:cond-on-m-main-prop} in Prop.~\ref{prop:main}. 
These requirements are met  if \eqref{eq:req-props2-3-on-homog-sets}
(or \eqref{eq:req-props2-3-on-struct-sets}) holds with $\eta = \delta \epsilon^2$, if $\cl K$ is homogeneous (resp. structured).

The astute reader could notice that if non-Gaussian but
  sub-Gaussian random matrices respect \eqref{eq:RIP12-embed-SLGW},
  \eg a Bernoulli random matrix with $\Phi_{ij} \sim_{\iid} \cl
  B(\{\pm 1\})$, the development above would show that the
  observations made in \cite{jacques2015small} are not tight. As
  announced above, \cite{jacques2015small} indeed proves that for such
  random matrices, the resulting quantized mapping respects the $(\cl D_{\ell_1},
\ell_2)$-QRIP$(\cl K^*, \epsilon, \rho)$ with $\rho = \delta^2
\epsilon + \rho_{\rm sg}$ and $\rho_{\rm sg} >0$ constant. However, \eqref{eq:RIP12-embed-SLGW} cannot hold in general for any sub-Gaussian random matrix since taking a Bernoulli $\bs \Phi$ and $\epsilon < |\gone - 1|$ in
\eqref{eq:RIP12-embed-SLGW} clearly fails as showed by the counter-example $\cl K \ni \bs v :=
(1,0,\cdots,0)^\top$ since $|\tinv{m} \gone\|\bs\Phi \bs v\|_1 - 1| = |\gone - 1| > \epsilon$.

\paragraph{(c) RSGE for multi-dithered quantized mappings:}
Still associated to the use of RSGE matrices, we can highlight a potential
\emph{multi-dithered} generalization of $\Amapd$ in
Prop.~\ref{prop:bi-dithered-embedding}. The non-linear map $\Amapd$
defined in \eqref{eq:bi-dithered-mapping} can indeed be extended over $\bb
R^{n \times d}$ to $d \geq 2$ and $\bs \Xi \sim \cl
U^{m\times d}([0, \delta])$, \ie associating each row of $\bs \Phi$
with $d$ different dithers. Following a similar proof to the one
developed for Prop.~\ref{prop:bi-dithered-embedding}, the independence of the $d$ columns of $\bs
\Xi$ could then
lead to a connection between matrices satisfying the
$(\ell_d,\ell_q)$-RIP for $d \geq 2$ and quantized random maps
respecting the $(\Demb,\ell_q)$-QRIP with $\Demb = \cl
D_\circ$. This is the case of matrices drawn according to random
Gaussian ensembles,  which 
are known to respect $(\ell_d,\ell_2)$-RIP$(\Sigma_{s,n}, \epslin)$
provided $m \geq m_0$ with $m_0 = O((s \log n/s)^{d/2})$
\cite{jacques2011dequantizing}. However, as there is no other known
matrix construction satisfying this RIP, we
prefer to not investigate here this potential generalization.

\paragraph{(d) Structured random matrix constructions (SRMC):} 
As explained in the introduction, there is a growing literature interested in the
development of fast quantized embeddings. Many works focus for
instance on the possibility to define fast 1-bit embeddings by
replacing $\cl Q$ with a sign operator without considering a
pre-quantization dithering
\cite{oymak2016near,yu2015binary,yu2014circulant}. However, to the best of our knowledge, such fast
1-bit embeddings are currently available for finite sets only, with
sometimes strong restrictions between their cardinality, the ambient
dimension $n$ and the embedding
dimension $m$ (\eg $m < n^{1/2}$ in~\cite{oymak2016near}). 

This paper aims at showing that if we can afford a different
quantization process, \ie a simple uniform quantization, a
suitable dither allows us to leverage the now large
market of random matrices known
to respect the RIP. This includes, for instance, structured random matrix constructions (SRMC)
with fast vector encoding schemes such as random orthonormal basis
ensembles~\cite{rudelson2008sparse,foucart2013mathematical},
random convolutions~\cite{rauhut2012restricted,romberg2009compressive}, 
spread-spectrum~\cite{puy2012universal} and scrambled block-Hadamard ensembles~\cite{gan2008fast}. Moreover, some of those
SRMC known to only respect the\footnote{Actually, a
  multiresolution version of the RIP is satisfied by most of those SRMC~\cite{oymak2015isometric}.}
$(\ell_2,\ell_2)$-RIP$(\Sigma_{s,n}, \epslin)$,
generally provided $m = O(s \log(n)^{O(1)} \log(s)^{O(1)})$, can be
extended to the embedding of general sets with small Gaussian
 mean width by a simple random sign flipping of the encoded
vector~\cite{oymak2015isometric}, \ie similarly to the effect of
spread-spectrum~\cite{puy2012universal}.

Finally, let us mention that the adjacency matrices of high-quality
unbalanced expander graphs have been shown to provide
$(\ell_p,\ell_p)$-RIP$(\Sigma_{s,n}, \epslin)$ matrix constructions
with $1\leq p \leq 1 + O(1/\log n)$ and fast vector encoding
schemes. Combined with Prop.~\ref{prop:main}, such matrices determine
a fast quantized map satisfying, \whp, the $(\cl D_{\ell_1},
\ell_1)$-QRIP$(\Sigma_{s,n}, \epslin, c\delta\epsilon)$ for some
$c>0$, provided $m \gtrsim \epsilon^{-2} s \log \frac{n}{s} \log \frac{1}{\delta \epsilon^2}$.  
 
\paragraph{(e) Fast 1-bit quantized embeddings of low-complexity sets:} Combined with Prop.~\ref{prop:main}, Prop.~\ref{prop:distorted-main} and
Prop. \ref{prop:bi-dithered-embedding}, the SRMC above can thus provide \whp fast
quantized embeddings of low-complexity sets of $\bb R^n$, \eg with
log-linear vector encoding time, provided $m$
respects the respective requirements of these propositions. 

Remarkably, this also allows the design of fast 1-bit quantized embedding
of low-complexity sets.  Indeed, for bounded sets $\cl K$ and
sufficiently large quantization resolution $\delta$ each component of the resulting
quantized map $\Amapd$ can be essentially encodable using
only one bit per measurement. 
For instance, if $\bs \Phi$ respects the $(\ell_p,\ell_q)$-RIP$(\cl K, \epslin)$ defined in~\eqref{eq:RIPpq} for some $\epslin \in (0,1)$ and integers $p,q \geq 1$, a rough computation yields that for all $\bs x \in \cl K$, $\|\bs \Phi \bs x\|_{\infty} \leq \|\bs \Phi \bs x\|_{p} \leq (1+\epslin)^{1/p} \|\bs x\|_q$. Therefore, setting $\delta$ larger than $2^{1/p} \|\cl K\|_q$ guarantees that $\Amap$ in~\eqref{eq:quant-embed-def} can only take two values per component.
In other words, the dithered quantizer $\cl Q$ will act as
a dithered sign operator in this resolution regime.

Notice that in such a 1-bit regime, it is interesting to compare Prop.~\ref{prop:main} and
Prop. \ref{prop:bi-dithered-embedding} with the binary embedding described in \cite{plan2014dimension}. This work shows that if $\cl K \subset \bb S^{n-1}$ and provided $m \gtrsim \epsilon^{-6} w(\cl K)^2$, with probability exceeding $1-2\exp(-c\epsilon^2 m)$ for some universal constant $c>0$, we can generate from $\bs \Phi \in \cl N^{m\times n}(0,1)$ a \emph{$\epsilon$-uniform tessellation} from the 1-bit mapping $\Amap_{\rm bin}(\cdot) := \sign(\bs \Phi\,\cdot)$, \ie using our notations
\begin{equation}
  \label{eq:bin-tessel}
  \ts  \cl D_{\rm ang}(\bs x, \bs x') - \epsilon\ \leq\ \cl D_{\rm H}(\Amap_{\rm bin}(\bs x), \Amap_{\rm bin}(\bs x))\ \leq\ \cl D_{\rm ang}(\bs x, \bs x')+ \epsilon,\quad \forall \bs x,\bs x' \in \cl K, 
\end{equation}
where $\cl D_{\rm ang}(\bs x, \bs x') := \tfrac{2}{\pi} \arccos(\bs x^\top\bs x) \in [0,1]$ is the normalized angular distance between $\bs x, \bs x' \in \bb S^{n-1}$, and $\cl D_{\rm H}(\bs w, \bs w') := \tinv{m} \sum_{i} 1_{w_i \neq w_j}$ is the normalized Hamming distance between two binary strings $\bs w, \bs w' \in \{\pm 1\}^{m}$ (with $1_{w_i \neq w_j} = 1$ if $w_i \neq w_j$ and 0 otherwise).  

We can directly observe some similarities between this property and the QRIP defined in~\eqref{eq:QRIP-def}. First, the Hamming distance is similar to $\cl D_{\cl \ell_1}$ considered in Prop.~\ref{prop:main} since $\cl D_{\rm H}(\bs w, \bs w') = 2^{-p} \cl D_{\ell_p}(\bs w, \bs w')$ for all $p\geq 1$ and $\bs w, \bs w' \in \{\pm 1\}^m$. Second, \eqref{eq:bin-tessel} shows that one can approximate the angular distance separating unit vectors in the subset $\cl K$ with the Hamming distance of their binary images obtained by $\Amap_{\rm bin}$, and the requirement on $m$ for this to happen is similar to the requirements shown in Prop.~\ref{prop:main} and Prop.~\ref{prop:bi-dithered-embedding} for ensuring that $\Amap$ respects a form of the QRIP. In fact, using Sudakov's inequality in \eqref{eq:Sudakov-minoration} and following the developments made in the points (a) and (b) above, for $\cl K^* = \cl K \cap \bb B^n$ and $\cl K$ homogeneous, the conditions of these propositions are respected if $\bs \Phi$ satisfies the proper RIP on $\cl K - \cl K$ and if  
\begin{equation}
  \label{eq:cond-m-bin-map-hom}
\ts m \gtrsim \frac{1}{\epsilon^6\delta^2}\, w(\cl K^*)^2,
\end{equation}
while for the structured sets described in Sec.~\ref{Main-results}, we will meet these conditions if 
\begin{equation}
  \label{eq:cond-m-bin-map-struct}
\ts m \gtrsim \frac{1}{\epsilon^2}\, w(\cl K^*)^2\log(1 + \frac{1}{\epsilon^3\delta}).  
\end{equation}
Interestingly, the condition \eqref{eq:cond-m-bin-map-hom} has the same dependency in $\epsilon^{-6}$ than the condition guaranteeing \eqref{eq:bin-tessel}, while \eqref{eq:cond-m-bin-map-struct} has a reduced dependency in $\epsilon^{-2}$ that is similar to the one ensuring the binary embedding of sparse vectors, as described in \cite{jacques2013robust}.  

However, while \eqref{eq:bin-tessel} displays no multiplicative distortion but only an additive one controlled by $\epsilon>0$, this binary embedding concerns only the approximation of angular distances of unit vectors in $\cl K \subset \bb S^{n-1}$. 
Even if $\delta$ is sufficiently large for $\Amap$ to become essentially binary for some bounded set $\cl K \subset \bb R^n$, the quantized embeddings proposed in Prop.~\ref{prop:main} and Prop.~\ref{prop:bi-dithered-embedding} still approximate the $\ell_q$-distance between vectors in $\cl K$.

Finally, note that Prop.~\ref{prop:distorted-main} may not be compared with~\eqref{eq:bin-tessel} since the additive distortion term in $\delta \|\bs x - \bs x'\|$ that appears in~\eqref{eq:QRIP-2nd-result} cannot be made arbitrarily small when~$\epsilon$ decays\footnote{An exception to this could be shown for $\epsilon$-close vectors $\bs x, \bs x'$.}.

\paragraph{(f) Asymmetric embeddings and quantized rank-one projections:} Our developments
can be easily extended to the use of asymmetric linear embeddings, as
defined by the asymmetric $(\ell_p, \ell_q)$-RIP$(\cl K, \epslin; C_1, C_2)$ 
\begin{equation}
  \label{eq:RIPpq-asym}
  (C_1-\epslin) \|\bs x\|^p_q \leq \tfrac{\mu_{\Phi}}{m} \|\bs \Phi \bs
  x\|^p_p \leq  (C_2+\epslin) \|\bs x\|^p_q,\quad \forall \bs x \in \cl K,
\end{equation}
for which the two bounds $0 < C_1 < C_2$ are not necessarily close to
each other.  

Let us illustrate this with the following corollary that describes
a particular asymmetric form of QRIP. It is a direct
consequence of both Prop.~\ref{prop:main} and the interesting
example of random Rank-One Projections (ROP) of matrices \cite{cai2015rop,kueng2015low,puy2015recipes}.

\begin{corollary}[Quantized rank-1 embedding of bounded low-rank matrices]
\label{cor:qrip-rank-one}
Let $\cl X$ be a centered sub-Gaussian distribution with unit variance. There exist two constants $C_1 \in (0,1/3)$ and $C_2 > 1$ depending only on $\cl X$, and three universal constants $C,c,c'>0$, such that, for $\epsilon \in (0,1)$, a quantization resolution $\delta > 0$, and $r,n_1,n_2 \in \bb N^0$, if 
$$
\ts m \gtrsim \epsilon^{-2}\,r (n_1 + n_2) \log\big(1 + \tfrac{1}{\min(1,\delta) \epsilon^2}\big),
$$
then, for $2m+1$ random vectors $\bs a_i \sim \cl X^{n_1}$, $\bs b_i \sim \cl X^{n_2}$ and $\bs \xi \sim \cl U^m([0,\delta])$ (with $i \in [m]$), $\kappa := 2/(C_1+C_2)$, and with probability exceeding $1-C \exp(-c \epsilon^2 m)$, the \emph{quantized} rank-1 measurement mapping 
\begin{equation}
  \label{eq:quant-rop}
  \ts \Amap :\ \bs U \in \bb R^{n_1 \times n_2}\ \mapsto\ \Amap(\bs U) := \big\{\kappa^{-1} \cl Q(\kappa\,\bs a_i^\top \bs U \bs b_i + \xi_i): i\in [m]\big\} \in\
  \bb R^m,  
\end{equation}
with $\cl Q$ defined in \eqref{eq:quant-def}, satisfies the asymmetric $(\cl D_{\ell_1}, \ell_F)$-QRIP 
$$
\ts (C_1 - \epsilon) \|\bs U - \bs V\|_F - c'\delta\epsilon\ \leq\ \frac{1}{m}\|\Amap(\bs U) - \Amap(\bs V)\|_1\ \leq\ 
(C_2 + \epsilon) \|\bs U - \bs V\|_F + c'\delta\epsilon,
$$
for all $\bs U, \bs V \in \cl R_r \cap \bb B_{\ell_F}^{n_1 \times n_2}$.
\end{corollary}

This corollary is directly connected to the remarkable properties of
random ROP, which are defined by the linear random mapping
\begin{equation}
  \label{eq:rand-rank1-linop}
\ts \cl A: \bs U \in \bb R^{n_1 \times n_2}\ \mapsto\ \cl A(\bs U) = \big\{\cl A_i(\bs U) := \bs a_i^\top \bs U \bs b_i: i \in [m]\big\} \in \bb R^m,
\end{equation}
where $\bs a_i \in \bb R^{n_1}$ and $\bs b_i \in \bb R^{n_2}$ are $2m$
random vectors. Each measurement $\cl A_i$ is rank-1 in the
sense that $\cl A_i(\bs U) = \scp{\bs a_i \bs
  b_i^\top}{\bs U}_F$ with the rank-1 \emph{probing} matrix $\bs a_i \bs
  b_i^\top \in \bb R^{n_1 \times n_2}$. 

Recovering low-complexity matrices (\eg low-rank) from their ROP has many interesting applications including recommender systems, phase retrieval and quadratic signal sensing, linear system identification, blind signal deconvolution, blind system calibration, as well as quantum state tomography (see, \eg \cite{cai2015rop, ling2015self, cambareri2016through,gross2011recovering} and \cite[Sec. 7]{davenport2016overview}).  
 
This recovery is possible from the following embedding property. Given $\epslin \in (0,1)$, it is established\footnote{Up to a minor modification of the proof of \cite[Prop. 7.1]{cai2015rop} available in the supplementary material
of the same paper.} in \cite[Prop. 7.1]{cai2015rop} (see also \cite[Sec. IV.B]{puy2015recipes} for a
similar result) that if the components of the random vectors $\bs a_i, \bs b_i$ are \iid according to
a centered sub-Gaussian distribution $\cl X$ with unit variance, and if 
\begin{equation}
  \label{eq:cond-m-rank1}
  \ts m \gtrsim \frac{1}{\epslin^2} r(n_1 + n_2) \log(\frac{1}{\epslin}),  
\end{equation}
then $\cl A$ is shown to respect \whp the Restricted Uniform
Boundedness (RUB) of order $r$, \ie 
there exist two constants $C_1 \in (0,1/3)$ and $C_2>1$ depending only on the sub-Gaussian norm of $\cl X$,
such that, with probability exceeding $1 - C\exp(-c\epslin^2 m)$ for
some $C,c>0$, 
\begin{equation}
  \label{eq:rub-def}
\ts  (C_1 - \epslin) \|\bs U\|_F \leq \frac{1}{m}\|\cl A(\bs U)\|_1
\leq (C_2 + \epslin) \|\bs U\|_F,
\end{equation}
for all rank-$r$ matrices $\bs U \in \bb R^{n_1 \times n_2}$, \ie $\bs
U \in \cl R_r$. In fact, if $\cl A$ respects the RUB of order $\beta r$
with $C_2/C_1 < \sqrt \beta$ for some $\beta \geq 2$, then one can recover any $\bs U_0 \in \cl R_r$ from $\cl A(\bs U_0)$ by a constrained nuclear norm minimization method \cite{cai2015rop}. 

We can observe that the relation \eqref{eq:rub-def} is a clear linear embedding of
rank-$r$ matrices into $(\bb R^m, \ell_1)$, \ie up to a rescaling of $\cl A$ by $\kappa :=
2/(C_1+C_2)$, it is a $(\ell_1,\ell_F)$-RIP$(\cl R_r, \epslin')$
with $\epslin' = (C_2 - C_1 + 2\epslin)/(C_1+C_2)$. However, this
embedding is characterized by a non-negligible and constant multiplicative distortion since it is easily shown that $\epslin' > \frac{3}{5}$. 

Interestingly, by identifying the $\ell_2$-norm with
the Frobenius $\ell_F$-norm and given $\epsilon \in (0,1)$, Prop.~\ref{prop:main} can anyway
leverage the RUB in \eqref{eq:rub-def} in order to define, with
probability exceeding $1-C \exp(-c \epsilon^2 m)$, a $(\cl D_{\ell_1},
\ell_F)$-QRIP$(\cl R_r\cap \bb B_{\ell_F}^{n_1 \times n_2}, \epslin',
c'\delta \epsilon)$ (for some $C,c,c'>0$) from the mapping~\eqref{eq:quant-embed-def} where
$\bs \Phi$ is replaced by $\cl A$. This is sustained by the fact that $\cl A$ can also be rewritten as the action of a $m \times n$ matrix $\bs \Phi$ (with $n=n_1n_2$) on the vectorization of $n_1 \times n_2$ matrices, the $i^{\rm th}$ row of $\bs \Phi$ being associated to the vectorization of~$\bs a_i \bs  b_i^\top$.   

From~\eqref{eq:cond-on-m-main-prop} and the bound on $\cl H(\cl
R_r\cap \bb B_{\ell_F}^{n_1 \times n_2}, \eta)$ given in
Tab.~\ref{tab:Kolmog-bounds}, this is allowed if
\begin{equation}
  \label{eq:m-qrip-rank-1}
\ts m \gtrsim \epsilon^{-2}\,r (n_1 + n_2) \log(1 + \tfrac{1}{\delta \epsilon^2}).
  \end{equation}
Notice that if we set $\epsilon = \epslin$, \eqref{eq:m-qrip-rank-1} is similar to \eqref{eq:cond-m-rank1} up to a dependence in $\log(1/\delta)$, and the probabilities for which \eqref{eq:rub-def} and the QRIP hold have then also the same trend.
This context provides Cor.~\ref{cor:qrip-rank-one} and its requirements by first using a union bound over the probability of the RUB above and the one provided by Prop.~\ref{prop:main}, and second by returning to the initial scaling of~$\cl A$ (which explains the use of $\kappa^{-1}$ in \eqref{eq:quant-rop}).

\subsection{Connections with the Representation and Coding of Signal Geometry}
\label{sec:comp-with-other}

Removing the condition on $\rho$, the QRIP defined in
\eqref{eq:QRIP-def} can be seen as a particular instance of the
(distorted) embedding concept introduced in
\cite{boufounos2015representation}. In this work, a non-linear map ${\sf F}:\cl K \to \cl E$ is a $(g,\epslin,\rho)$-embedding of $(\cl K,\cl D_{\cl K})$ in
$(\cl E, \Demb)$ if, for some invertible $g:\bb R \to \bb R$, we have 
\begin{equation}
  \label{eq:petros-embedding}
  \ts(1-\epslin)g(\cl D_{\cl K}(\bs u,\bs u')) - \rho \leq \Demb({\sf
  F}(\bs u), {\sf
  F}(\bs u')) \leq (1+\epslin)g(\cl D_{\cl K}(\bs u,\bs u')) + \rho,
\end{equation}
for all $\bs u,\bs u' \in \cl K$. From this definition, a
$(\Demb,\ell_q)$-QRIP$(\cl K, \epslin, \rho)$ map $\Amapd$ is thus a
$((\cdot)^\pE, \epslin,\rho)$-embedding of $(\cl K,\cl D_{\cl K})$ in
$(\cl E, \Demb)$ with the additional requirement~\eqref{eq:cond-rho-qrip} on $\rho$. 

In~\cite{boufounos2015representation} 
the authors show that, under certain conditions, if ${\sf F}$ is \whp a $(g, \epslin,\rho)$-embedding of $(\cl K,\cl D_{\cl K})$ in
$(\cl E, \Demb)$ \emph{on an arbitrary pair of vectors} $\bs u,\bs u' \in \cl
K$, then, for some $c>1$ depending on {\sf
  F} and $g$, ${\sf F}$ is also
\whp a
$(g, \epslin,c\rho)$-embedding of the whole set $\cl K$. The
conditions for this to happen are that we can control both the size of any covering of $\cl K$ (\ie
with its Kolmogorov entropy) and the probability $P_T$ to have $T$ discontinuity points of ${\sf F}$ in any
balls of radius~$r$, \ie in order to show that $P_T$ quickly decays if $r$ decreases and $T$ increases. 

From this general result, the authors are then able to study the
alteration of the random linear map $\bs \Phi$ (whose rows are \iid as some suitable vector distribution) defined through
$\bs x \mapsto \Amapd_f(\bs x) = f(\bs \Phi \bs x + \bs \xi)$ for some
$1$-periodic function $f$ (possibly discontinuous) and with
$\bs \xi \sim \cl U^{m}([0,1])$. In particular the function $f$,
through its Fourier series, and the distribution of the rows of
$\bs \Phi$ determine together an explicit
``kernelization'' of the set $\cl K$, \ie it enables a geometric
encoding of $\cl K$ where small pairwise vector distances in $\cl K$
are better encoded than large ones.  

By considering a uniform quantizer $\cl Q$ as a discontinuous map, it could seem
that the procedure described in~\cite{boufounos2015representation}, and in particular in Theorem
3.2 therein, could be used to prove our results on specific random matrix
constructions (\eg RSGE). However, even if it was
possible, we believe that such an adaptation
would not be straightforward. 

The standpoint of~\cite{boufounos2015representation} is
indeed fully probabilistic, \ie in~\eqref{eq:petros-embedding} the stochastic behavior of ${\sf
  F}$ is characterized from the set $\cl K$ by the combined action
a random matrix $\bs \Phi$ and of a random dither $\bs \xi$. In our work, we prefer
to assume (an instance of) the RIP of the
linear map $\bs \Phi \in \bb R^{m \times n}$ to first embed $\cl K$
in $\bb R^m$ with an appropriate $\ell_p$-distance, and then
characterize the embedding of $\bb R^m$ in the quantized domain $\bb
Z_\delta^{m}$ thanks to our dithered quantization. As explained above,
this trick allows
us to directly integrate, for our specific case, a much broader class of random matrix
constructions, with possibly log-linear vector encoding time.
 
Another difference between~\cite{boufounos2015representation} and our
work comes from the treatment reserved to the discontinuities of the
non-linear map. As explained above, in
\cite{boufounos2015representation} the authors need to
characterize the Lipschitz continuity of $f$ ``by parts'', \ie between its
discontinuity points. Then, the
probability $P_T$ of having $T$ discontinuities in the neighborhood of
radius $r$ of any $\bs x \in \cl K$, \ie that exactly $T$ discontinuity frontiers will
pass through this neighborhood, must be bounded and proved to decay rapidly when $r$ decreases and $T$
increases. This is indeed important to get a very general framework valid for
any non-linear map ${\sf F}$. Therefore, considering a $r$-covering of
$\cl K$ of cardinality $C_r$ with balls of radius
$r$ (whose cardinality is controlled by the Kolmogorov entropy), the
slicing of all those balls in these $T$ parts related to the
discontinuities of $f$ is still a $r$-covering of $\cl K$. Taking arbitrarily one point per ball
slice thus defines another $r$-net of $\cl K$ with cardinality at most
$T C_r$, on which
\eqref{eq:petros-embedding} can hold \whp by union bound, which
provides an easy extension to the whole set $\cl K$ by continuity. A
side effect of this elegant and general analysis is, however, the
necessity to control $P_T$.  

Our approach benefits of the choice of $f=\cl Q$. In our case, we
can soften the distance used in the quantized
embedding domain $\cl E = \bb Z_\delta^m$ in order to avoid any
estimation of $P_T$. As described in the different proofs provided in
Secs.~\ref{proof-of-main-prop}-\ref{sec:bi-dith-quant}, we can
then directly analyze the continuity of this new distance within a small
neighborhood of the image by $\bs \Phi$ of an $\eta$-net of $\cl K$
(with $\eta$ depending on $\epsilon$ and $\delta$), whose radius can
be estimated from the RIP of $\bs \Phi$.   
   
An interesting open problem would consist in the study of a general softening
procedure of the discontinuities induced by any non-linear
function $f$ in $\Demb$ through a map  $\Amapd_f$. We could then
analyze if fast linear encoding schemes are compatible with more
advanced non-linear maps, where $\cl Q$ in~\eqref{eq:quant-embed-def}
or in~\eqref{eq:bi-dithered-embedding} is replaced by, \eg non-uniform quantization
\cite{gray1998quantization,jacques2013stabilizing}, vector or binned
quantization~\cite{nguyen2010frame,pai2006nonadaptive}, \emph{non-regular}
quantization \cite{boufounos2012universal,boufounos2015representation}, or by more
general periodic functions as in~\cite{boufounos2015representation,boufounos2015universal}.

\section{Proof of Proposition~\ref{prop:main}}
\label{proof-of-main-prop}

Prop.~\ref{prop:main} is formally demonstrated by showing
first that we can embed, with high probability, $(\cl J :=\bs \Phi(\cl K - \cl K),
\ell_1)$ in $(\cl Q(\cl J + \bs \xi), \ell_1)$ for $\bs \xi \sim \cl U^m([0,\delta])$, and then using the
linear $(\ell_1,\ell_q)$-embedding allowed by $\bs \Phi$ (through the
assumed RIP) to relate $(\cl J :=\bs \Phi(\cl K - \cl K),
\ell_1)$ to $(\cl K - \cl K,
\ell_q)$.    

As explained in the proof sketch provided after Prop.~\ref{prop:main}, the first step of the proof is obtained by
softening the (pseudo) distance $\cl D_{\ell_1}(\cl Q(\bs a), \cl Q(\bs
a')) = \tinv{m} \|\cl Q(\bs a) - \cl Q(\bs
a')\|_1$, \ie the distance from which $\cl D_{\ell_1}(\Amap(\bs x), \Amap(\bs x'))$ is
derived by taking $\bs a = \bs \Phi \bs x + \bs \xi$ and $\bs a' = \bs \Phi
\bs x' + \bs \xi$ with $\Amapd$ defined
in~\eqref{eq:quant-embed-def}. 

According to a procedure defined in~\cite{jacques2015small} that we 
rewrite here for completeness, this is
done by softening each of the $m$ elements composing the sum
$m\cl D_{\ell_1}(\cl Q(\bs a), \cl Q(\bs
a')) = \sum_{i=1}^m |\cl Q(a_i) - \cl Q(a'_i)|$. Let us first observe
that, for $a, a' \in \bb R$, 
$$
\ts d(a,a') := |\cl Q(a) - \cl Q(a')| := \delta \sum_{k\in \bb Z} \bb I_{\cl S}( a - k\delta, a'
- k\delta), 
$$
with $\cl S = \{(a,a') \in \bb R^2: aa' < 0\}$ and $\bb I_{\cl C}(a,a')$
is the indicator of $\cl C$ evaluated in $(a, a')$, i.e., it is equal to 1 if $(a,a') \in \cl C$ and 0 otherwise. In words, $d$
counts the number of \emph{thresholds} in $\delta \bb Z$ that can be
inserted between $a$ and $a'$, \ie we can alternatively write $d(a,a') =
\delta |(\delta \bb Z) \cap [a,a']|$. 

Introducing the
set $\cl S^t = \{(a,a') \in \bb R^2: a < - t, a' > t\} \cup \{(a,a') \in
\bb R^2: a >  t, a' < - t\}$ for $t\in \bb R$, we can define a soft
version of $d$ by
\begin{equation}
  \label{eq:def-d-t}
  \ts d^t(a,a') := \delta \sum_{k\in \bb Z} \bb I_{\cl S^t}( a - k\delta, a'
  - k\delta).   
\end{equation}
From the definition of $\cl S^t$, the value of $t$ determines a set of forbidden (or relaxed) intervals $\delta \bb Z + [-|t|, |t|] = \{[k\delta-|t|, k\delta+|t|]: k \in \bb Z\}$ if $t > 0$ (resp. $t < 0$) of size $2|t|$ and centered on the quantizer thresholds in $\delta \bb Z$. For $t >0$ a threshold of $\delta \bb Z$ is not counted in $d^t(a,a')$ if $a$ or $a'$ fall in its forbidden interval, whereas for $t<0$ a threshold that is not between $a$ and $a'$ can be counted if $a$ or $a'$ fall inside its relaxed interval. 

This new (pseudo) distance is clearly a decreasing function of $t$. In fact, for $t\geq 0$,
$$
d^{t}(a,a') \leq d(a, a') \leq d^{-t}(a, a')
$$ 
since $\cl S^{t}\subset \cl S \subset \cl S^{-t}$. 
Moreover, we can show as in~\cite{jacques2015small} that 
\begin{align}
\label{eq:diff-dt-ds}
\big|d^t(a, a') - d^s(a, a')\big|&\leq\ 4(\delta + |t-s|),\\
\label{eq:diff-dt-dist}
\big|d^t(a, a') - |a - a'|\big|&\leq\ 4(\delta + |t|),\\   
\label{eq:ex-ected-dither-diff-dt-dist}
\big|\bb E d^t(a + \xi, a' + \xi) - |a - a'|\big|&\leq\ 4|t|,    
\end{align}
where $\xi \sim \cl U([0, \delta])$ and the last inequality is proved in~\cite[Appendix
C]{jacques2015small}. 

As expressed in the next lemma, the distance $d^t$ already displays a certain form of continuity that
does not hold for $d$.
\begin{lemma}
\label{lem:cont-smald-t}
For all $a,a' \in \bb R$, $t\in \bb R$ and $|r|,|r'|\leq \epsilon$, we have 
\begin{equation}
  \label{eq:dt-cont-prop}
  d^{t+\epsilon}(a,a') \leq d^{t}(a+r,a'+r') \leq d^{t-\epsilon}(a,a').
\end{equation}
\end{lemma}
\begin{proof}
Considering the definition of the set $\cl S^t$ in
\eqref{eq:def-d-t}, we have
$$
\cl S^{t + \epsilon}(a, a') \subset 
\cl S^{t}(a+r, a'+r') \subset 
\cl S^{t - \epsilon}(a, a').
$$
Indeed, for instance $a+r < -t$ involves $a \leq - t +
\epsilon$, while $a+r > t$ involves $a \geq t -
\epsilon$, and similarly for all other conditions involved in the
definition of $\cl S^{t + \epsilon}$, $\cl S^{t}$ and $\cl S^{t - \epsilon}$. 
Considering these inclusions the result follows since for any indicator $\bb I_{\cl C} \leq \bb I_{\cl
  C'}$ if $\cl C \subset \cl C'$.
\end{proof}
From the scalar (pseudo) distance $d^t$, we can thus define for $\bs a, \bs a' \in
\bb R^m$, 
$$
\ts \cl D^t(\bs a, \bs a') = \tinv{m} \sum_{i} d^t(a_i,a'_i),
$$
with $\cl D^0(\bs a, \bs a') = \cl D_{\ell_1}\big(\cl Q(\bs a), \cl
Q(\bs a')\big)$. This distance inherits from $d^t$ the following
continuity property.  

\begin{lemma}[Continuity of $\cl D^t$ with respect to $\ell_p$-perturbations] 
\label{lem:continuity-Lq-Dt-new} 
Let $p \geq 1$ and $\bs a,\bs a',\bs r,\bs r' \in \bb R^m$. We assume
  that $\max(\|\bs r\|_p,\|\bs r'\|_p) \leq \eta \sqrt[p]{m}$ for some $\eta >0$. Then for every $t\in \bb R$ and
  $P\geq 1$ one has
  \begin{equation}
    \label{eq:continuity-L1-Dt}
    \ts \cl D^{t+\eta \sqrt[p]{P}}(\bs a,\bs a') - 8 (P^{-1}\delta +
    \eta) \leq 
    \cl D^t(\bs a + \bs r,\bs a' + \bs r') \leq 
    \cl D^{t-\eta \sqrt[p]{P}}(\bs a,\bs a') + 8 (P^{-1}\delta +
    \eta).
  \end{equation}  
\end{lemma}
\begin{proof}
The proof is inspired by the proof of Lemma 5.5 in
\cite{plan2014dimension} valid for 1-bit (sign) quantization, and an adaptation of~\cite[Lemma
3]{jacques2015small} to $\ell_p$-perturbations. By assumption, the set 
$$
T := \{i \in [m]: |r_i| \leq \eta \sqrt[p]{P},\, |r'_i| \leq \eta \sqrt[p]{P}\}
$$
is such that $|T^{\compl}| \leq 2 m/P$ as $2\eta^p m \geq \|\bs r\|^p_p + \|\bs r'\|^p_p \geq \|(\bs r)_{T}\|^p_p + \|(\bs r')_{T}\|^p_p + |T^\compl|\eta^p P \geq
|T^\compl|\eta^p P$. Therefore, considering
Lemma~\ref{lem:cont-smald-t}, we find by simple inequalities and with 
$\rho_i := \max(|r_i|,|r'_i|)$,
\begin{align*}
&\cl D^{t+\eta \sqrt[p]{P}}(\bs a,\bs a') =\ts \tfrac{1}{m} \sum_{i=1}^m
  d^{t+\eta \sqrt[p]{P}}(a_i, a'_i)\\
&\leq \ts \tfrac{1}{m} \sum_{i\in T}
  d^{t+\eta \sqrt[p]{P}}(a_i, a'_i) + 
  \tfrac{1}{m} \sum_{i\in T^\compl}
  d^{t+\eta \sqrt[p]{P}}(a_i, a'_i)\\  
&\leq \ts \tfrac{1}{m} \sum_{i\in T}
  d^{t}(a_i + r_i, a'_i + r'_i) + 
  \tfrac{1}{m} \sum_{i\in T^\compl}
  d^{t}(a_i + r_i, a'_i + r'_i)\\  
&+ \ts
  \tfrac{1}{m} \sum_{i\in T^\compl}
  \big[ d^{t+\eta \sqrt[p]{P} - \rho_i}(a_i + r_i, a'_i + r'_i) - d^{t}(a_i + r_i, a'_i + r'_i)\big].
\end{align*}
Using~\eqref{eq:diff-dt-ds} to bound the last sum of the last
expression and since, by definition of $T$, $\rho_i \geq \eta \sqrt[p]{P}$ for $i\in T^\compl$, we find 
\begin{align*}
\cl D^{t+\eta \sqrt[p]{P}}(\bs a,\bs a')&\textstyle \leq \cl D^{t}(\bs
                                       a + \bs r,\bs a' + \bs r') +
                                       \tfrac{4}{m} \sum_{i\in
                                       T^\compl} (\delta + \rho_i - \eta \sqrt[p]{P})\\
&\textstyle \leq \cl D^{t}(\bs a + \bs r,\bs a' + \bs r') +
\tfrac{8\delta}{P} + \tfrac{4}{m}\sum_{i\in T^\compl} \rho_i.
\end{align*}
However, 
$$
\ts \sum_{i\in T^\compl} \rho_i \leq \|\bs
r_{T^\compl}\|_1 + \|\bs
r'_{T^\compl}\|_1\leq m^{1-1/p}\|\bs
r\|_p + m^{1-1/p}\|\bs
r'\|_p \leq \,2\eta m,
$$
so that
\begin{align*}
\cl D^{t+\eta \sqrt[p]{P}}(\bs a,\bs a')
&\textstyle \leq \cl D^{t}(\bs a + \bs r,\bs a' + \bs r') +
\tfrac{8\delta}{P} + 8 \eta,
\end{align*}
which provides the lower bound of~\eqref{eq:continuity-L1-Dt}.
For the upper bound, 
\begin{align*}
&\cl D^{t-\eta \sqrt[p]{P}}(\bs a,\bs a') =\ts \tfrac{1}{m} \sum_{i=1}^m d^{t-\eta \sqrt[p]{P}}(a_i,a'_i)\\
&\geq \ts \tfrac{1}{m} \sum_{i\in T} d^{t}(a_i + r_i, a'_i
  + r'_i) + \tfrac{1}{m} \sum_{i\in T^\compl} d^{t-\eta
  \sqrt[p]{P}+\rho_i}(a_i + r_i, a'_i + r'_i)\\  
&\geq \ts \cl D^{t}(\bs a + \bs r,\bs a' + \bs r') - \big|\tfrac{1}{m}
  \sum_{i\in T^\compl} \big(d^{t-\eta \sqrt[p]{P}+\rho_i}(a_i + r_i,a'_i + r'_i) - d^{t}(a_i + r_i, a'_i + r'_i)\big)\big|,
\end{align*}
and the last sum is bounded as above.  
\end{proof}

Before delving into the proof of Prop.~\ref{prop:main} and following
its aforementioned proof sketch, we now need to
study how the random
variable $\cl D^t(\bs a + \bs \xi, \bs a' + \bs \xi)$, for $\bs \xi
\sim \cl U^m([0, \delta])$, concentrates around its expected value,
and how it deviates from $\cl D_{\ell_1}(\bs a, \bs a')$. 
\begin{lemma}
\label{lem:non-uniform-concent-soften-dist}
Let $\bs \xi \sim \cl U^m([0, \delta])$ and $\bs a, \bs a' \in \bb
R^m$. We have 
\begin{align}
  \label{eq:soft-dist-concent}
  \bb P_{\bs \xi}[ \big | \cl D^t(\bs a + \bs \xi, \bs a' + \bs \xi) -
  \cl D_{\ell_1}(\bs a, \bs a') \big | > 4|t| + \epsilon (\delta + |t|) ]\quad&\lesssim\quad e^{-c \epsilon^2 m}.\\
  \label{eq:bound-on-expect-dist-t}
  \bb |\bb E \cl D^t(\bs a + \bs \xi, \bs a' + \bs \xi) - \cl
    D_{\ell_1}(\bs a, \bs a')|\quad&\lesssim\quad |t|.
\end{align}
\end{lemma}
\begin{proof}
Most of the elements of this proof are already detailed in
\cite{jacques2015small} but we prefer to rewrite them here briefly under our
specific notations.  Denoting $Z^t_i := d^t(a_i + \xi_i, a'_i + \xi_i)
- |a_i - a'_i|$,~\eqref{eq:diff-dt-dist} provides $\|Z^t_i\|_{\psi_2}
\leq \|Z^t_i\|_{\infty} \lesssim \delta + |t|$, which shows that each
\rv $Z^t_i$ is sub-Gaussian. Therefore, by the approximate rotation invariance of sub-Gaussian
variables~\cite{vershynin2010introduction}, $Z^t := \sum_i Z^t_i =
m \cl D^t(\bs a + \bs \xi, \bs a' + \bs \xi) - \|\bs a - \bs a'\|_1$ is
sub-Gaussian with norm 
$$
\ts \|Z^t\|^2_{\psi_2} \lesssim m(\delta + |t|)^2.
$$ 
Consequently, there is a
$c>0$ such that $\bb P(|Z^t - \bb E Z^t| > \epsilon) \leq e\,\exp(-c
\epsilon^2/\|Z^t\|^2_{\psi_2})$, so that
$$
\ts \bb P( |\tinv{m}Z^t| > |\tinv{m}\bb E Z^t| + \epsilon (\delta + |t|)) \leq \bb P( |\tinv{m}(Z^t - \bb E Z^t)| > \epsilon (\delta + |t|))
\lesssim \exp(- c\epsilon^2m),
$$
which provides~\eqref{eq:soft-dist-concent} since from~\eqref{eq:ex-ected-dither-diff-dt-dist}
$$
|\tinv{m}\bb E Z^t| = |\bb E\cl D^t(\bs a + \bs \xi, \bs a' + \bs \xi) -
\tinv{m} \|\bs a - \bs a'\|_1| \leq 4 |t|,
$$
as written in~\eqref{eq:bound-on-expect-dist-t}.   
\end{proof}

\begin{proof}[Proof of Prop.~\ref{prop:main}]

Let us define an $(\ell_q,\eta)$-covering $\cl K_\eta$ of the set $\cl K$ with
$\log |\cl K_\eta| \leq \cl H_q(\cl K, \eta)$. We fix $\epslin, \epsilon \in (0,1)$ and $t\in \bb R$, then assume that $\bs
\Phi$ satisfies the $(\ell_1,\ell_q)$-RIP$(\cl K - \cl K, \epslin)$. Therefore,
for all $\bs x, \bs x' \in \cl K$ with their respective closest vectors
in $\cl K_\eta$ being $\bs x_0$ and $\bs x'_0$, we have from~\eqref{eq:RIPpq}
\begin{equation}
  \label{eq:bounded-image-resid}
\cl D_{\ell_1}(\bs \Phi \bs x, \bs \Phi \bs x_0)\leq m\,(1+\epslin) \|\bs x
- \bs x_0\|_q \leq 2m\eta\quad \text{and}\quad \cl D_{\ell_1}(\bs \Phi \bs x', \bs \Phi \bs x'_0) \leq 2m \eta.  
\end{equation}

Moreover, since $\bs \Phi$ is fixed, we observe from a union bound
applied on~\eqref{eq:soft-dist-concent} that if $m \gtrsim \epsilon^{-2} \cl H_q(\cl K,
\eta)$ then, for $P \geq 1$ to be fixed later, both relations
\begin{align}
  \label{eq:concent-cover-upper-2}
| \cl D^{t-\eta P} (\bs \Phi \bs x_0 + \bs \xi, \bs \Phi \bs x'_0 + \bs \xi) -
  \cl D_{\ell_1}(\bs \Phi \bs x_0, \bs \Phi \bs x'_0) \big |&\leq
4|t|+4\eta P +\epsilon(\delta+ |t| +\eta P),\\  
  \label{eq:concent-cover-lower-2}
| \cl D^{t+\eta P}(\bs \Phi \bs x_0 + \bs \xi, \bs \Phi \bs x'_0 + \bs \xi) -
  \cl D_{\ell_1}(\bs \Phi \bs x_0, \bs \Phi \bs x'_0) \big |&\leq 
4|t|+4\eta P +\epsilon(\delta+ |t| +\eta P),
\end{align}
hold jointly for all $\bs x_0, \bs x'_0$ in $\cl
K_\eta$ (\ie on no more than $|\cl K_{\eta}|^2 \leq \exp(2 \cl H_q(\cl
K, \eta))$ vectors) with probability exceeding $1 - C e^{-c \epsilon^2 m}$ for some
$C,c>0$. 
  
Consequently, since for any $\bs x, \bs x' \in \cl K$, we can
always write $\bs x = \bs x_0 + \bs r$ and $\bs x' = \bs
x_0' + \bs r'$ with $\bs x_0, \bs x_0' \in \cl K_{\eta}$
and $\max(\|\bs r\|_q, \|\bs r'\|_q) \leq \eta$, using Lemma
\ref{lem:continuity-Lq-Dt-new} with $p=1$,
\eqref{eq:concent-cover-upper-2} and again
\eqref{eq:bounded-image-resid}, we have with the same probability and
for some $c>0$, 
\begin{align*}
\cl D^{t} (\bs \Phi \bs x + \bs \xi, \bs \Phi \bs x' + \bs \xi)&\leq \cl D^{t-\eta P} (\bs \Phi \bs x_0 + \bs \xi, \bs \Phi \bs x_0' + \bs \xi)   
+ 8 (P^{-1} \delta + \eta)\\
&\leq \cl D_{\ell_1}(\bs \Phi \bs x_0, \bs \Phi \bs x'_0) + 
4|t|+4\eta P +\epsilon(\delta+ |t| +\eta P) + 8 (P^{-1} \delta + \eta)\\
&\leq \cl D_{\ell_1}(\bs \Phi \bs x, \bs \Phi \bs x') + 
4|t|+2\eta (2 P + 1) +\epsilon(\delta+ |t| +\eta P) + 8 (P^{-1} \delta + \eta)\\
&\leq \cl D_{\ell_1}(\bs \Phi \bs x, \bs \Phi \bs x') + c(|t| +
  \delta\epsilon)\\
&\leq (1+\epslin) \|\bs x - \bs x'\|_q + c(|t| +
  \delta\epsilon),
\end{align*}
where, for some $c>0$, we set the free parameters to
$P^{-1} = \epsilon$ and $\eta = \delta \epsilon^{2} < \delta\epsilon$ giving $\eta P =
\delta \epsilon$. The lower bound is obtained
similarly using~\eqref{eq:concent-cover-lower-2}, and
Prop.~\ref{prop:main} is finally obtained with $t=0$.
\end{proof}

\section{Proof of Proposition~\ref{prop:distorted-main}}
\label{sec:proof-distorted-embedding}

The proof of this proposition is essentially similar to the one of
Prop.~\ref{prop:main}.  We formally show first that we can embed, with high probability, $(\cl J :=\bs \Phi(\cl K - \cl K),
\ell_2)$ in $(\cl Q(\cl J + \bs \xi), \ell_2)$ for $\bs \xi \sim \cl U^m([0,\delta])$, and then, using the
linear $(\ell_2,\ell_2)$-embedding allowed by $\bs \Phi$ (through the
RIP), we relate $(\cl J :=\bs \Phi(\cl K - \cl K),
\ell_2)$ to $(\cl K - \cl K,
\ell_2)$. However, some subtle yet important differences come from the quadratic nature of
$\Demb = \cl D_{\ell_2}$. First, conversely to $\cl D_{\cl E} = \cl
D_{\ell_1}$ where $\bb E \cl D_{\ell_1}(\Amap(\bs x), \Amap(\bs x)) =
\|\bs \Phi(\bs x-\bs x')\|_1$, in the quadratic case $\bb E \cl
D_{\ell_2}(\Amap(\bs x), \Amap(\bs x)) \neq \|\bs \Phi(\bs x-\bs x')\|^2$, \ie a
distortion is already present in expectation. Second, the
concentration of $\cl D_{\ell_2}(\Amap(\bs x), \Amap(\bs x))$ around its
mean is not independent of $\|\bs \Phi(\bs x-\bs x')\|$ anymore.  And finally,
the mandatory softening of the pseudo-distance $\cl D_{\ell_2}(\Amap(\bs x), \Amap(\bs
x))$, as imposed to reobtain a certain form of continuity in the
neighborhoods of $\bs x$ and $\bs x'$, also depends on
$\|\bs \Phi(\bs x - \bs x')\|$. We will see below that these three effects worsen the additive
distortion of the $(\ell_2,\ell_2)$-quantized embedding.

We start by considering the general soft (pseudo) distance 
$$
\ts \cl D_2^t(\bs a, \bs a') := \tinv{m} \sum_{i} (d^t(a_i,a'_i))^2,
$$
for $\bs a, \bs a' \in \bb R^m$ and $d^t$ defined as before.
We thus have $\cl D_2^0(\bs a, \bs a') = \cl D_{\ell_2}\big(\cl Q(\bs
a), \cl Q(\bs a')\big)$.

As for $\cl D^t$ in the previous section, $\cl D^t_2$ can be shown to
respect a certain form of continuity under $\ell_2$-perturbations, \ie
those induced by the image under $\bs \Phi$ of any
$\ell_2$-perturbation of a vector in $\cl K$ around its image thanks
to the $(\ell_2,\ell_2)$-RIP. 
\begin{lemma}[Continuity of $\cl D_2^t$ with respect to $\ell_2$-perturbations] 
\label{lem:continuity-L2-Dt2} 
Let $\bs a,\bs a',\bs r,\bs r' \in \bb R^m$. We assume
  that $\max(\|\bs r\|,\|\bs r'\|) \leq \eta \sqrt{m}$ for some $\eta >0$. Then for every $t\in \bb R$ and
  $P\geq 1$ one has
  \begin{align}
    \label{eq:continuity-L2-Dt2-upper}
    \cl D^t_2(\bs a + \bs r,\bs a' + \bs r') - \ts \cl D_2^{t-\eta
    \sqrt{P}}(\bs a,\bs a')&\lesssim \tfrac{\eta\sqrt P + \delta}{\sqrt {m P}} \|\bs a - \bs
                            a'\| + (\delta + |t| + \eta)
                            \eta + \delta (\delta + |t|)\, \tfrac{2}{P},\\
    \label{eq:continuity-L2-Dt2-lower}
    \ts \cl D_2^{t+\eta \sqrt{P}}(\bs a,\bs a') - \cl D^t_2(\bs a +
    \bs r,\bs a' + \bs r')&\lesssim \tfrac{\eta\sqrt P + \delta}{\sqrt {m P}} \|\bs a - \bs
                            a'\| + (\delta + |t| + \eta)
                            \eta + \delta (\delta + |t|)\, \tfrac{2}{P}.
  \end{align}
\end{lemma}
\begin{proof}
The result is achieved by simply taking $\bs A = \bs a \bs 1_2^\transp$, $\bs A' = \bs a' \bs
1_2^\transp$, $\bs R = \bs r \bs 1_2^\transp$ and $\bs R' = \bs r' \bs
1_2^\transp$ in the more general Lemma
\ref{lem:continuity-L2-Dt2-matrix-version} given in
Sec.~\ref{sec:bi-dith-quant}.  
\end{proof}

The question is now to study the concentration properties of $\cl D_2^t(\bs a + \bs \xi,
\bs a' + \bs \xi)$ for $\bs \xi \sim \cl U^m([0,\delta])$. We first
observe how $(d^t)^2$ behaves under a scalar dithering of its inputs. 
\begin{lemma}
\label{lem:expect-ellp-soften-dist-bounds}
Let $\xi \sim \cl U([0, \delta])$ and $a, a' \in \bb
R$. For $|a-a'|
  \leq \delta - 2|t|$, we have 
\begin{equation}
\label{eq:tight-bound-on-expect-ellp-dist-t-small-dist}
\ts  \delta (|a-a'| -
4|t|)_+\ \leq\ \bb
E\,\big(d^t(a + \xi, a' + \xi)\big)^2\ \leq\ \delta (|a-a'|
+ 4|t|),
\end{equation}
while, regardless of the value of $|a-a'|$,  
\begin{align}
  \label{eq:lower-bound-on-expect-ellp-dist-t}
&\bb E\, \big(d^t(a + \xi, a' + \xi)\big)^2 \ \geq\ \big(|a - a'| - 4|t|\big)_+^2\\
  \label{eq:upper-bound-on-expect-ellp-dist-t}
&\bb E\, \big(d^t(a + \xi, a' + \xi)\big)^2\ \leq\ \big(|a - a'| + 4(\delta +
|t|)\big)^{2} (|a-a'| + 4|t|)\big).
\end{align}
\end{lemma}
\begin{proof}
Notice that if $|a-a'|\leq \delta - 2|t|$ for $a,a' \in \bb R$, which
implicitly imposes $|t|\leq \delta/2$, $X := \delta^{-1}
d^t(a+\xi,a' + \xi) \in \{0, 1\}$, so that 
\begin{align*}
&\ts \bb E (d^t(a + \xi, a' + \xi))^2\ =\  \delta\,\bb E d^t(a +
  \xi, a' + \xi)\ \in\ [\delta (|a-a'| - 4|t|)_+, \delta (|a-a'| + 4|t|)].  
\end{align*}
By the Jensen inequality and using~\eqref{eq:ex-ected-dither-diff-dt-dist} we find also
\begin{align*}
&\ts \bb E (d^t(a + \xi, a' + \xi))^2\geq (\bb E d^t(a +
  \xi, a' + \xi))^2 \geq (|a-a'| - 4|t|)_+^2.  
\end{align*}
For this last relation, we notice from~\eqref{eq:diff-dt-dist} that $0 \leq X \leq \delta^{-1}(|a -
a'| + 4(\delta + |t|))$, so that using again
\eqref{eq:ex-ected-dither-diff-dt-dist} we get
$$
\bb E\, \big(d^t(a + \xi, a' + \xi)\big)^2\ \leq\
\big(|a - a'| + 4(\delta +
|t|)\big)\,(|a-a'| + 4|t|).
$$
\end{proof}

Denoting $\overline{\cl D}^t_2(\bs a, \bs a') := \bb E\,\cl D_2^t(\bs a + \bs \xi,
\bs a' + \bs \xi)$ with $\bs \xi \sim \cl U([0, \delta])$, we can now determine how
$\overline{\cl D}^t_2(\bs a, \bs a')$ deviates from $\|\bs a -\bs a'\|^2$.
\begin{corollary}
Let $\bs \xi \sim \cl U^m([0,\delta])$. Then, for $\bs a,\bs a' \in
\bb R^m$,
\begin{align}
\label{eq:big-sq-Dt-mean-low-bounds}
\ts \overline{\cl D}^t_2(\bs a, \bs a') - \cl D_{\ell_2} (\bs a,\bs a')&\ts \geq - 8|t|\,\cl D ^{1/2}_{\ell_2} (\bs a,\bs
a') + 16|t|^2,\\
\label{eq:big-sq-Dt-mean-up-bounds}
\ts \overline{\cl D}^t_2(\bs a, \bs a') - \cl D_{\ell_2} (\bs a,\bs a') 
                                                                     &\ts \leq (8|t| + 4\delta)\,\cl D ^{1/2}_{\ell_2} (\bs a,\bs
a') + 4|t|\,(4|t|+\delta),
\end{align}
which provides the loose bound 
\begin{align}
\label{eq:big-sq-Dt-mean-loose-bounds}
\ts |\overline{\cl D}^t_2(\bs a, \bs a') - \cl D_{\ell_2} (\bs a,\bs a')| 
                                                                     &\ts \leq (8|t| + 4\delta)\,\cl D ^{1/2}_{\ell_2} (\bs a,\bs
a') + 4|t|\,(4|t|+\delta).
\end{align}
\end{corollary}
\begin{proof}
For any $a,a' \in \bb R$,
Lemma~\ref{lem:expect-ellp-soften-dist-bounds} provides
$$
(|a - a'| - 4|t|)_+^2\ \leq\ \bb E\, \big(d^t(a + \xi, a' + \xi)\big)^2\
\leq\ (|a - a'| + 4|t|)^2 + 4\,\delta\, (|a-a'| + 4|t|).
$$
Therefore, 
\begin{align*}
\ts\sum_i \bb E\, \big(d^t(a_i + \xi_i, a' + \xi_i)\big)^2&\geq \ts \sum_i (|a_i - a'_i| - 4|t|)_+^2\\
&\geq\ \|\bs a - \bs
a'\|_2^2 + 16 m |t|^2 - 4|t| \|\bs a - \bs a'\|_1\\
&\geq\ \|\bs a - \bs
a'\|_2^2 + 16 m |t|^2 - 4|t| \sqrt m\,\|\bs a - \bs a'\|_2\\
&\ts = m\,(\,\tinv{\sqrt m}\|\bs a - \bs
a'\|_2 - 4|t|)^2,  
\end{align*}
and
\begin{align*}
\ts\sum_i \bb E\, \big(d^t(a_i + \xi_i, a' + \xi_i)\big)^2&\leq m\,(\,\tinv{\sqrt m}\|\bs a - \bs
a'\|_2 + 4|t|)^2 + 4\delta \|\bs a - \bs a'\|_1 + 4\delta|t|m\\
&\leq m\,(\,\tinv{\sqrt m}\|\bs a - \bs
a'\|_2 + 4|t|)^2 + 4\delta \sqrt{m} \|\bs a - \bs a'\|_2 + 4\delta|t|m.
\end{align*}
\end{proof}

The last step before proving Prop.~\ref{prop:distorted-main} is to
study the concentration properties of the \rv $\cl D_2^t(\bs a + \bs \xi, \bs a' + \bs \xi)$.
\begin{lemma}
\label{lem:non-uniform-concent-ellp-soften-dist}
Given $\bs \xi \sim \cl U^m([0, \delta])$ and $\bs a, \bs a' \in \bb
R^m$, there exist two constants $c,c'>0$ such that 
\begin{align}
  \label{eq:soft-ellp-dist-concent}
\ts \bb P\big[ |\cl D_2^t(\bs a + \bs \xi, \bs a' + \bs \xi) -
  \cl D_{\ell_2}(\bs a,\bs a')| > c'(\delta + |t|) \cl D^{1/2}_{\ell_2}(\bs a,\bs a') +
  c'\delta^2\epsilon + |t|(\delta + |t|)\,\big]&\lesssim e^{- c\epsilon^2m}.
\end{align}
\end{lemma}
\begin{proof}
Denoting $Z^t_i := (d^t(a_i + \xi_i, a'_i + \xi_i))^2
- |a_i - a'_i|^2$, we find from~\eqref{eq:diff-dt-dist}
\begin{align*}
|Z^t_i|&= \big(d^t(a_i + \xi_i, a'_i + \xi_i) + |a_i - a'_i|\big)\,\big| d^t(a_i + \xi_i, a'_i + \xi_i)
- |a_i - a'_i|\big|\\
&\lesssim (\delta + |t|) (|a_i - a'_i| + \delta + |t|),   
\end{align*}
which provides $\|Z^t_i\|_{\psi_2} \leq \|Z^t_i\|_{\infty} \lesssim
(\delta + |t|) (|a_i - a'_i| + \delta + |t|)$ and proves the
sub-Gaussianity of $Z^t_i$. Therefore,  $Z^t := \sum_i Z^t_i =
m (\cl D_2^t(\bs a + \bs \xi, \bs a' + \bs \xi))^2 - \|\bs a - \bs a'\|^2_2$ is
sub-Gaussian with norm 
\begin{align*}
\ts \|Z^t\|^2_{\psi_2}&\ts \lesssim (\delta + |t|)^2 \sum_i (|a_i - a'_i| +
\delta + |t|)^2\\
&\ts = (\delta + |t|)^2 \big( \|\bs a - \bs a'\|^2 + 2
(\delta + |t|) \|\bs a - \bs a'\|_1 + (\delta + |t|)^2 m\big)\\
&\ts \lesssim  m (\delta + |t|)^2 \big(\tinv{\sqrt m}\|\bs a - \bs a'\|_2 +
  (\delta + |t|)\big)^2,\\
&\ts \lesssim  m (\delta + |t|)^2 \big(\cl D^{1/2}_{\ell_2}(\bs a,\bs a') +
  (\delta + |t|)\big)^2,
\end{align*}
from the approximate rotation invariance of sub-Gaussian
variables~\cite{vershynin2010introduction}. Consequently, there is a
$c>0$ such that $\bb P(|Z^t - \bb E Z^t| > \epsilon) \lesssim \exp(-c
\epsilon^2/\|Z^t\|^2_{\psi_2})$, which provides
\begin{align*}
&\ts \bb P\big[ |\tinv{m}Z^t| > |\tinv{m}\bb E Z^t| + \epsilon (\delta + |t|) \big(\cl D^{1/2}_{\ell_2}(\bs a,\bs a') +
  (\delta + |t|)\big)\big]\\
&\ts \leq \bb P\big[ |\tinv{m}(Z^t - \bb E Z^t)| > \epsilon (\delta + |t|) \big(\cl D^{1/2}_{\ell_2}(\bs a,\bs a') +
  (\delta + |t|)\big) \big]
\lesssim \exp(- c\epsilon^2m).
\end{align*}
However, from~\eqref{eq:big-sq-Dt-mean-low-bounds} and~\eqref{eq:big-sq-Dt-mean-up-bounds},
$$
|\tinv{m}\bb E Z^t| = |\overline{\cl D}^t_2(\bs a,\bs a') - \cl
D_{\ell_2}(\bs a, \bs a')| \lesssim \ts (|t| + \delta)\,\big(\cl D ^{1/2}_{\ell_2} (\bs a,\bs
a') + |t|\big).
$$
Therefore, up to a rescaling of $c>0$ and considering that $\epsilon <
1$, there is a $c'>0$ such that 
\begin{align*}
\ts \bb P\big[ |\cl D_2^t(\bs a + \bs \xi, \bs a' + \bs \xi) -
  \cl D_{\ell_2}(\bs a,\bs a')| > c'(\delta + |t|) \cl D^{1/2}_{\ell_2}(\bs a,\bs a') +
  c'\delta^2\epsilon + |t|(\delta + |t|)\,\big]&\lesssim e^{- c\epsilon^2m}.
\end{align*}  
\end{proof}

\begin{proof}[Proof of Proposition~\ref{prop:distorted-main}]
As in the proof of Prop.~\ref{prop:main}, let us define an
$(\ell_2,\eta)$-covering $\cl K_\eta \subset \cl K$ of the set $\cl K$ with
$\log |\cl K_\eta| \leq \cl H_2(\cl K, \eta)$ and $\eta$ fixed
later. We fix $\epslin,\epsilon \in (0,1)$ and $t\in \bb R$ and assume that $\bs
\Phi$ satisfies the $(\ell_2,\ell_2)$-RIP$(\cl K - \cl K, \epslin)$,
\ie for all $\bs u, \bs u' \in \cl K$, 
\begin{equation}
  \label{eq:tmp-rip}
  \cl D_{\ell_2}(\bs \Phi\bs u, \bs \Phi\bs u') \leq (1+\epslin) \|\bs
  u - \bs u'\|^2,\quad \cl D^{1/2}_{\ell_2}(\bs \Phi\bs u, \bs \Phi\bs u') \leq \sqrt{2}\,\|\bs
  u - \bs u'\|.     
\end{equation}
Therefore,
for all $\bs x, \bs x' \in \cl K$ with their respective closest vectors $\bs x_0$ and $\bs x'_0$
in $\cl K_\eta$, we have also
\begin{equation}
  \label{eq:bounded-image-resid-rip2}
\max(\|\bs \Phi(\bs x - \bs x_0)\|_2, \|\bs \Phi(\bs x' - \bs
x'_0)\|_2) \leq \sqrt{2m}\,\eta.  
\end{equation}

Moreover, since $\bs \Phi$ is fixed, we observe from a union bound
applied on~\eqref{eq:soft-ellp-dist-concent} that if $m \gtrsim \epsilon^{-2} \cl H_2(\cl K,
\eta)$ then, for $P \geq 1$ to be fixed later, both relations
\begin{align}
&\ts\big|\cl D_2^{t-\eta \sqrt{P}} (\bs \Phi \bs x_0 + \bs \xi, \bs \Phi \bs x'_0 + \bs \xi) -
  \cl D_{\ell_2}(\bs \Phi \bs x_0,\bs \Phi \bs x'_0)\big|\nonumber\\
&\ts\qquad\lesssim (\delta + |t-\eta \sqrt{P}|) \cl D^{1/2}_{\ell_2}(\bs \Phi \bs x_0,\bs \Phi \bs x'_0) +
  \delta^2\epsilon + |t-\eta \sqrt{P}|(\delta + |t-\eta \sqrt{P}|)\nonumber\\ 
\label{eq:concent-cover-upper-3}
&\ts\qquad \lesssim (\delta + |t| +\eta \sqrt{P}) \|\bs x_0
  - \bs x'_0\| +
  \delta^2\epsilon + (|t| + \eta \sqrt{P})\,(\delta + |t| + \eta
  \sqrt{P})
\end{align}
and
\begin{align}
&\ts\big|\cl D_2^{t+\eta \sqrt{P}} (\bs \Phi \bs x_0 + \bs \xi, \bs \Phi \bs x'_0 + \bs \xi) -
  \cl D_{\ell_2}(\bs \Phi \bs x_0,\bs \Phi \bs x'_0)\big|\nonumber\\
\label{eq:concent-cover-lower}
&\ts\qquad \lesssim (\delta + |t| +\eta \sqrt{P}) \|\bs x_0
  - \bs x'_0\| +
  \delta^2\epsilon + (|t| + \eta \sqrt{P})\,(\delta + |t| + \eta \sqrt{P})
\end{align}
hold jointly for all $\bs x_0, \bs x'_0$ in $\cl
K_\eta$ with probability exceeding $1 - c e^{-c' \epsilon^2 m}$ for some
$c,c'>0$. 
  
Consequently, for any vectors $\bs x = \bs x_0 + \bs r, \bs x' = \bs
x_0' + \bs r'$ in $\cl K$ with $\bs x_0, \bs x_0' \in \cl K_{\eta}$
and $\max(\|\bs r\|, \|\bs r'\|) \leq \eta$, we have $\max(\|\bs \Phi\bs r\|, \|\bs \Phi\bs
r'\|) \leq \sqrt {2m}\,\eta$ from
\eqref{eq:bounded-image-resid-rip2}. Lemma
\ref{lem:continuity-L2-Dt2} and~\eqref{eq:continuity-L2-Dt2-upper}
give then
\begin{align*}
&\cl D^t_2(\bs \Phi \bs x + \bs \xi,\bs \Phi \bs x' + \bs \xi) - \cl D_2^{t-\eta \sqrt{P}}(\bs \Phi \bs x_0 + \bs \xi,\bs \Phi \bs x'_0 + \bs \xi)\\
&\ts \lesssim \tfrac{\eta + \delta}{\sqrt m} \|\bs\Phi\bs x_0 - \bs\Phi\bs x'_0\| + (\delta + |t| + \eta)
                            \eta + \delta (\delta + |t|)\,
  \tfrac{2}{P},\\
&\ts \lesssim (\eta + \delta) \|\bs x_0 - \bs x'_0\| + (\delta + |t| + \eta)
                            \eta + \delta (\delta + |t|)\,
  \tfrac{2}{P},
\end{align*}
so that, using~\eqref{eq:concent-cover-upper-3} we find
\begin{align*}
&\cl D^t_2(\bs \Phi \bs x + \bs \xi,\bs \Phi \bs x' + \bs \xi) - \cl D_{\ell_2}(\bs \Phi \bs x_0,\bs \Phi \bs x'_0)\\ 
&\lesssim (\delta + |t| +\eta \sqrt{P}) \|\bs x_0
  - \bs x'_0\| +
  \delta^2\epsilon + (|t| + \eta \sqrt{P})\,(\delta + |t| + \eta \sqrt{P})\\ 
&\ts\qquad + (\eta + \delta) \|\bs x_0 - \bs x'_0\| + (\delta + |t| + \eta)
                            \eta + \delta (\delta + |t|)\,
  \tfrac{2}{P}\\
&\lesssim (\delta + |t| +\eta \sqrt{P}) \|\bs x_0
  - \bs x'_0\| +
  \delta^2(\epsilon + \tfrac{2}{P}) + (|t| + \eta \sqrt{P})\,(\delta + |t| + \eta
  \sqrt{P}) + \delta |t| \tfrac{2}{P}\\
&\lesssim (\delta + |t|) \|\bs x_0
  - \bs x'_0\| +
  ( \delta\epsilon + |t|)(\delta + |t|)\\
&\lesssim (\delta + |t|) \|\bs x
  - \bs x'\| + 
  ( \delta\epsilon + |t|)(\delta + |t|),
\end{align*}
where we set the free parameters to
$P^{-1} = \epsilon$ and $\eta = \delta \epsilon^{3/2} <
\delta\epsilon$, giving $\eta \sqrt{P} =
\delta \epsilon$. Finally, from~\eqref{eq:tmp-rip} we have that $\cl D_{\ell_2}(\bs \Phi \bs x_0,\bs \Phi \bs x'_0) \leq (1+\epslin) \|\bs x_0 - \bs
x'_0\|^2$ while the covering provides $\|\bs x_0 - \bs
x'_0\|^2 \geq \|\bs x - \bs
x'\|^2 - 4\eta \|\bs x - \bs
x'\| + 4\eta^2$. Therefore,
\begin{align*}
&\cl D^t_2(\bs \Phi \bs x + \bs \xi,\bs \Phi \bs x' + \bs \xi) - (1+\epslin)\|\bs x
  - \bs x'\|^2 \lesssim (\delta + |t|) \|\bs x
  - \bs x'\| + 
  ( \delta\epsilon + |t|)(\delta + |t|).
\end{align*}
The lower bound is obtained
similarly using~\eqref{eq:concent-cover-lower}.
\end{proof}

\section{Proof of Proposition~\ref{prop:bi-dithered-embedding}}
\label{sec:bi-dith-quant}

As for the proofs of Props.~\ref{prop:main} and
\ref{prop:distorted-main}, proving
Prop. \ref{prop:bi-dithered-embedding} requires the following softened distance: let us define, for $\bs A := (\bs a_1, \bs a_2), \bs A' := (\bs a'_1,
\bs a'_2) \in \bb R^{m \times 2}$, the pre-metric\footnote{That is such
  that $\cl D^{t}(\bs A,\bs A') \geq 0$ and $\cl D^{t}(\bs A,\bs A)=0$.}   
\begin{equation}
  \label{eq:matrix-distance}
  \cl D^{t}(\bs A,\bs A') =\ts \tfrac{1}{m} \sum_{i=1}^m d^{t}(a_{i1}, a'_{i1}) d^{t}(a_{i2}, a'_{i2}).
\end{equation}

We directly observe that $\cl D^{0}(\bs A,\bs A') = \cl D_\circ(\cl
Q(\bs A), \cl
Q(\bs A'))$ so that we will be able to set later $\bs A = \bs \Phi\bs x \bs
1_2^\transp + \bs \xi$ and $\bs A' = \bs \Phi\bs x' \bs
1_2^\transp + \bs \xi$ for some $\bs x,\bs x' \in \bb R^n$.

As shown in the next lemma (proved in
Appendix~\ref{sec:proof-lemma-5}), this pre-metric displays some form of
continuity with respect to $\ell_2$-perturbations of its arguments. 
\begin{lemma}
\label{lem:continuity-L2-Dt2-matrix-version} 
Given $\bs A,\bs A',\bs R,\bs R' \in \bb R^{m\times 2}$, we assume
  that $\max(\|\bs r_1\|,\|\bs r_2\|,\|\bs r'_1\|,\|\bs r'_2\|) \leq
  \eta \sqrt{m}$ for some $\eta >0$, with $\bs r_i$ and $\bs r_i'$ the
  $i\th$-column of $\bs R$ and $\bs R'$, respectively. Then for every $t\in \bb R$ and
  $P\geq 1$ one has
  \begin{align}
    \label{eq:continuity-L2-Dt2-upper-matrix-version}
\cl D^{t}(\bs A + \bs R,\bs A' + \bs  R') - \cl D^{t-\eta
    \sqrt{P}}(\bs A,\bs A')&\ts\ \lesssim\ \cl L(\bs A - \bs A'),\\
    \label{eq:continuity-L2-Dt2-lower-matrix-version}
\cl D^{t+\eta \sqrt{P}}(\bs A,\bs A') - \cl D^{t}(\bs
                                       A + \bs R,\bs A' + \bs
  R')&\ts\ \lesssim\ \cl L(\bs A - \bs A').
  \end{align}
with
$$
\cl L(\bs B) := (\eta + \tfrac{\delta}{\sqrt P}) \tinv{\sqrt m} \|\bs B\|_F + (\delta + |t| + \eta)
  \eta + \delta (\delta + |t|)\, \tfrac{2}{P}.
$$
\end{lemma}

Moreover, when for some some $\bs a,\bs a' \in \bb R^n$, $\bs A = \bs a\bs
1_2^\transp + \bs \Xi$ and $\bs A' = \bs a'\bs
1_2^\transp+ \bs \Xi$ for $\bs \Xi \sim \cl U^{m\times
  2}([0,\delta])$,  the pre-metric $\cl D^{t}(\bs A,\bs A')$ taken as
a random variable of $\bs \Xi$ has an expected value on $t=0$ that
matches $\tinv{m}\|\bs a - \bs a'\|^2$, \ie an effect that was absent
from the quantized embedding defined by~\eqref{eq:quant-embed-def} as
explained in Sec.~\ref{Main-results} and Sec.~\ref{sec:proof-distorted-embedding}.
As a consequence, we can show that $\cl D^{t}(\bs A,\bs A')$ concentrates
around $\tinv{m}\|\bs a - \bs a'\|^2$ with reduced distortion. 
\begin{lemma}
\label{lem:concent-Pi-dist-around-mean}
Given $\bs a, \bs a' \in \bb R^{m}$, $\bs A = \bs a\bs
1_2^\transp + \bs \Xi$ and $\bs A' = \bs a'\bs
1_2^\transp+ \bs \Xi$ for $\bs \Xi \sim \cl U^{m\times
  2}([0,\delta])$, we have for some $C,c>0$ and probability smaller than
$C\exp(-c \epsilon^2 m)$, 
\begin{equation}
  \label{eq:lem:concent-Pi-dist-around-mean}
  \big|\cl D^{t}(\bs A,\bs A') - \tinv{m}\|\bs a - \bs
  a'\|^2\big| \geq (\epsilon
  (\delta + |t|) + |t|) \tfrac{8}{\sqrt m}\|\bs a - \bs a'\|  + 16\, \epsilon\, (\delta + |t|)^2 + 16|t|^2.
\end{equation}
\end{lemma}
\begin{proof}
We start by noting that we can rewrite
$$
\ts \cl D^{t}(\bs A,\bs A') - \tinv{m}\|\bs a - \bs
  a'\|^2 = \tinv{m}\,\sum_i (X^t_i Y^t_i - |a_i - a'_i|^2), 
$$  
with $X^t_i = d^t(a_{i} + \Xi_{i1}, a'_{i} + \Xi_{i1})$ and $Y^t_i =
d^t(a_{i} + \Xi_{i2}, a'_{i} + \Xi_{i2})$.  Moreover, from
\eqref{eq:ex-ected-dither-diff-dt-dist}, $\max(|\bb E X^t_i
- |a_{i} - a'_i||,|\bb E Y^t_i
- |a_i - a'_i||) \leq 4 |t|$, while~\eqref{eq:diff-dt-dist} gives
$\max(\|X^t_i - |a_i - a'_i|\|_\infty, \|Y^t_i - |a_i - a'_i|\|_\infty)
\leq 4(\delta + |t|)$.  
Consequently, we can apply Lemma~\ref{lem:concent-prod-vars} on the
concentration of the product of independent and bounded random
variables, the \rv's $X_i$ and $Y_i$ being both bounded by $L =
4(\delta + |t|)$. Thus, setting $s = 4|t|$ in this lemma gives the result.
\end{proof}

\begin{proof}[Proof of Prop.~\ref{prop:bi-dithered-embedding}]
As in the previous proofs of Prop.~\ref{prop:main} and Prop.~\ref{prop:distorted-main}, let us define an $(\ell_2,\eta)$-covering $\cl K_\eta$ of the set $\cl K$. We fix $\epslin,\epsilon \in (0,1)$ and $t\in \bb R$ and assume that $\bs
\Phi$ satisfies the $(\ell_2,\ell_2)$-RIP$(\cl K - \cl K, \epslin)$. Therefore,
for all $\bs x, \bs x' \in \cl K$ with their respective closest vectors
in $\cl K_\eta$ being $\bs x_0$ and $\bs x'_0$, we have 
\begin{equation}
  \label{eq:bounded-image-resid-l2}
\ts \|\bs \Phi(\bs x - \bs x_0)\|_2\leq \sqrt{2m}\,\eta\quad \text{and}\quad \|\bs
\Phi(\bs x' - \bs x'_0)\|_2 \leq \sqrt{2m}\,\eta,  
\end{equation}
with also $|\|\bs \Phi(\bs u - \bs u')\|^2_2 - \|\bs u - \bs u'\|^2|
\leq \epslin \|\bs u - \bs u'\|^2$ for all $\bs u,\bs u' \in \cl K$. 

Hereafter, we use again the compact notation $\bar{\bs u} = \bs u \bs
1_2^\transp$ for any vectors $\bs u$. Moreover, since $\bs \Phi$ is fixed, we observe from a union bound
applied on~\eqref{eq:lem:concent-Pi-dist-around-mean} that if $m \gtrsim \epsilon^{-2} \cl H_2(\cl K,
\eta)$ then, for $P \geq 1$ to be established later, both relations
\begin{align}
&\ts\big|\cl D^{t-\eta \sqrt{P}} (\bs \Phi \bar{\bs x}_0 +
  \bs \Xi, \bs \Phi \bar{\bs x}'_0  + \bs \Xi) -
  \cl D_{\ell_2}(\bs \Phi \bs x_0,\bs \Phi \bs x'_0)\big|\nonumber\\
&\ts\qquad\lesssim 
(\epsilon
  \delta + (1+\epsilon) |t-\eta \sqrt{P}|) \tfrac{8}{\sqrt m}\|\bs
  \Phi(\bs x_0 - \bs x'_0)\|\\
&\ts\qquad\qquad  + 16 \epsilon \delta^2 + 16|t-\eta \sqrt{P}|(2
  \epsilon \delta + (1+\epsilon) |t-\eta \sqrt{P}|) 
\nonumber\\ 
\label{eq:concent-cover-upper-bidith}
&\ts\qquad \lesssim 
(\epsilon
  \delta + |t| + \eta \sqrt{P}) \|\bs x_0 - \bs x'_0\|  + \epsilon \delta^2 + (|t| + \eta \sqrt{P})(
  \epsilon \delta + |t| + \eta \sqrt{P}) 
\\ 
&\ts\big|\cl D^{t+\eta \sqrt{P}} (\bs \Phi \bar{\bs x}_0 +
  \bs \Xi, \bs \Phi \bar{\bs x}'_0  + \bs \Xi) -
  \cl D_{\ell_2}(\bs \Phi \bs x_0,\bs \Phi \bs x'_0)\big|\\
\label{eq:concent-cover-lower-bidith}
&\ts\qquad \lesssim (\epsilon
  \delta + |t| + \eta \sqrt{P}) \|\bs x_0 - \bs x'_0\|  + \epsilon \delta^2 + (|t| + \eta \sqrt{P})(
  \epsilon \delta + |t| + \eta \sqrt{P}) 
\end{align}
hold jointly for all $\bs x_0, \bs x'_0$ in $\cl
K_\eta$ with probability exceeding $1 - c e^{-c' \epsilon^2 m}$ for some
$c,c'>0$. 
  
Consequently, as any pair of vectors $\bs x, \bs x' \in \cl K$ can always be
written as $\bs x = \bs x_0 + \bs r, \bs x' = \bs
x_0' + \bs r'$ in $\cl K$ with $\bs x_0, \bs x_0' \in \cl K_{\eta}$
and $\max(\|\bs r\|, \|\bs r'\|) \leq \eta$, using Lemma
\ref{lem:continuity-L2-Dt2-matrix-version} and~\eqref{eq:continuity-L2-Dt2-upper-matrix-version}, we find 
\begin{align*}
&\cl D^t(\bs \Phi \bar{\bs x} + \bs \Xi,\bs \Phi \bar{\bs x}' + \bs \Xi) - \cl D^{t-\eta \sqrt{P}}(\bs \Phi \bar{\bs x}_0 + \bs \Xi,\bs \Phi \bar{\bs x}'_0 + \bs \Xi)\\
&\ts \lesssim (\eta + \tfrac{\delta}{\sqrt P})\tinv{\sqrt m} \|\bs\Phi\bar{\bs x}_0 - \bs\Phi\bar{\bs x}'_0\|_F + (\delta + |t| + \eta)
                            \eta + \delta (\delta + |t|)\,
  \tfrac{2}{P},\\
&\ts \lesssim (\eta + \tfrac{\delta}{\sqrt P})\tinv{\sqrt m} \|\bs\Phi{\bs x}_0 - \bs\Phi{\bs x}'_0\| + (\delta + |t| + \eta)
                            \eta + \delta (\delta + |t|)\,
  \tfrac{2}{P},\\
&\ts \lesssim (\eta + \tfrac{\delta}{\sqrt P}) \|\bs x_0 - \bs x'_0\| + (\delta + |t| + \eta)
                            \eta + \delta (\delta + |t|)\,
  \tfrac{2}{P}\\
&\ts \lesssim (\eta + \tfrac{\delta}{\sqrt P}) \|\bs x - \bs x'\| + (\delta + |t| + \eta)
                            \eta + \delta (\delta + |t|)\,
  \tfrac{2}{P},
\end{align*}
where we used the $(\ell_2, \ell_2)$-RIP$(\cl K-\cl K, \epslin)$ and $\|\bs x_0 - \bs x_0'\|\leq \|\bs x - \bs x'\| +
2\eta$ from the covering of $\cl K_\eta$. Using again this last bound and~\eqref{eq:concent-cover-upper-bidith}, we get
\begin{align*}
&\cl D^t(\bs \Phi \bar{\bs x} + \bs \Xi,\bs \Phi \bar{\bs x}' + \bs \Xi) - \cl D_{\ell_2}(\bs \Phi \bs x_0,\bs \Phi \bs x'_0)\\ 
&\lesssim (\epsilon
  \delta + |t| + \eta \sqrt{P}) \|\bs x - \bs x'\|  + \epsilon \delta^2 + (|t| + \eta \sqrt{P})(
  \epsilon \delta + |t| + \eta \sqrt{P})\\
&\ts\qquad + (\eta + \tfrac{\delta}{\sqrt P}) \|\bs x - \bs x'\| + (\delta + |t| + \eta)
                            \eta + \delta (\delta + |t|)\,
  \tfrac{1}{P}\\
&\lesssim (\epsilon
  \delta + |t|) \|\bs x - \bs x'\|  +
  \epsilon\delta(|t| + \delta) + |t|^2,
\end{align*}
where we have set the free parameters to 
$P^{-1} = \epsilon^2$ and $\eta = \delta \epsilon^{2} < \delta\epsilon$ (with $\eta \sqrt{P} =
\delta \epsilon$). 

Finally, since $\cl D_{\ell_2}(\bs \Phi \bs x_0,\bs \Phi \bs x'_0) \leq (1+\epslin) \|\bs x_0 - \bs
x'_0\|^2$ and 
$$
\|\bs x_0 - \bs
x'_0\|^2 - \|\bs x - \bs
x'\|^2 \leq 2\eta \|\bs x - \bs
x'\| + 4\eta^2 \lesssim \delta\epsilon \|\bs x - \bs
x'\| + \delta^2 \epsilon,
$$ 
we find on $t=0$, \ie for $\cl D^0(\bs \Phi \bar{\bs x} + \bs \Xi,\bs \Phi \bar{\bs x}' + \bs \Xi)=\tinv{m}\| \Amapd(\bs x) - \Amapd (\bs
  x')\|_{1,\circ}$, that
\begin{align*}
&\tinv{m}\| \Amapd(\bs x) - \Amapd (\bs
  x')\|_{1,\circ} -
  (1+\epslin) \|\bs x - \bs x'\|^2\ \lesssim\ \epsilon
  \delta \|\bs x - \bs x'\|  +
  \epsilon\delta^2.
\end{align*}

The upper bound in \eqref{eq:bi-dithered-embedding} is then achieved by observing that
$2\epsilon\delta \|\bs x - \bs x'\| \leq \epsilon \|\bs x - \bs x'\|^2 + \epsilon\delta^2$ which then provides
\begin{align*}
&\tinv{m}\| \Amapd(\bs x) - \Amapd (\bs
  x')\|_{1,\circ} -
  (1+\epslin + c \epsilon) \|\bs x - \bs x'\|^2\ \lesssim\ \epsilon\delta^2,
\end{align*}
for some $c>0$.

The lower bound in
\eqref{eq:bi-dithered-embedding} is obtained
similarly using~\eqref{eq:concent-cover-lower-bidith}.  
\end{proof}

\appendix

\section{Concentration of the product of bounded and independent
  random variables}
\label{sec:useful-lemmas}

Lemma~\ref{lem:concent-Pi-dist-around-mean} relies on the following
useful property.
\begin{lemma}
\label{lem:concent-prod-vars} Let $X_i$ and $Y_i$ be independent and bounded random variables for
$1\leq i\leq m$ with $\bb E X_i = \bb E Y_i = \mu_i$. Assume that, for
some $\bs a \in \bb
R^m$, $|\mu_i - |a_i||\leq s$ and $\max(\|X_i -
|a_i|\|_{\infty}, \|Y_i -
|a_i|\|_{\infty}) \leq L$ for some $L,s > 0$. Then, for some $c>0$,
$$
\ts \bb P\big[\big|\tinv{m} \sum_i X_iY_i - |a_i|^2\big| \geq \epsilon L (L
+ \tfrac{2}{m}\|\bs a\|) + s(s + \tfrac{2}{m}\|\bs a\|)\big] \lesssim \exp(-c \epsilon^2 m).
$$      
\end{lemma}
\begin{proof}
Notice first that, for $Z_i := X_iY_i - |a_i|^2$ for $i\in [m]$, 
$$
\|Z_i\|_{\infty} \leq \|X_i(Y_i - |a_i|)\|_{\infty} +
\||a_i|\,(X_i - |a_i|)\|_{\infty} \leq L \|X_i\|_{\infty} + |a_i| L
\leq L^2 + 2|a_i| L.
$$  
Therefore, each $Z_i$ is sub-Gaussian with
$\|Z_i\|_{\psi_2} \leq \|Z_i\|_{\infty} \leq L (L + 2|a_i|)$. Moreover, by approximate rotational invariance
\cite{vershynin2010introduction}, 
\begin{align*}
\ts \|\sum_i Z_i\|^2_{\psi_2}&\ts \lesssim \sum_i \|Z_i\|^2_{\psi_2} \leq L^2
\sum_i (L^2 + 4|a_i| + 4|a_i|^2)\\
&\ts \leq L^2(m L^2 + 4 \sqrt m \|\bs
a\|_2 + 4 \|\bs a\|^2) \leq m L^2 (L + \tfrac{2}{m} \|\bs a\|)^2.   
\end{align*}
Consequently, for $\epsilon > 0$ and some $c>0$,
$$
\ts \bb P(|\sum_i Z_i - \bb E Z_i| \geq \epsilon) \lesssim \exp\big(-\epsilon^2 (m L^2 (L + \tfrac{2}{m} \|\bs a\|)^2)^{-1}\big).
$$
However, $|\bb E Z_i - |a_i|^2| \leq |\mu_i - |a_i|| |\mu_i + |a_i|| \leq
s (2|a_i| + s)$, so that $|\sum_i Z_i - \bb E Z_i| \geq |\sum_i Z_i -
|a_i|^2| - m s^2 - 2s \sqrt m \|\bs a\|$. Therefore
$$
\ts \bb P(|\sum_i Z_i - |a_i|^2| \geq \epsilon + m s^2 + 2s \sqrt m \|\bs a\|) \lesssim \exp\big(-\epsilon^2 (m L^2 (L + \tfrac{2}{m} \|\bs a\|)^2)^{-1}\big),
$$
which gives the result by a simple rescaling $\epsilon \to m \epsilon L (L + \tfrac{2}{m} \|\bs a\|)$.  
\end{proof}

\section{Proof of Lemma \ref{lem:continuity-L2-Dt2-matrix-version}}
\label{sec:proof-lemma-5}

By assumption, the sets 
$$
T_j := \{i \in [m]: |r_{ij}| \leq \eta \sqrt{P},\, |r'_{ij}| \leq \eta
\sqrt{P}\},\quad j \in \{1,2\},
$$
are such that $|T_j^{\compl}| \leq 2 m/P$ as $2\eta^2 m \geq \|\bs r_j\|^2_2 + \|\bs r'_j\|^2_2 \geq \|(\bs r_j)_{T}\|^2_2 + \|(\bs r'_j)_{T}\|^2_2 + |T^\compl|\eta^2 P \geq
|T^\compl|\eta^2 P$. Therefore, considering
Lemma~\ref{lem:cont-smald-t}, we find with 
$\rho_{ij} := \max(|r_{ij}|,|r'_{ij}|)$,
\begin{align*}
&\cl D^{t+\eta \sqrt{P}}(\bs A,\bs A') =\ts \tfrac{1}{m} \sum_{i=1}^m
  d^{t+\eta \sqrt{P}}(a_{i1}, a'_{i1}) d^{t+\eta \sqrt{P}}(a_{i2}, a'_{i2})\\
&= \ts \tfrac{1}{m} \sum_{i\in T}
  d^{t+\eta \sqrt{P}}(a_{i1}, a'_{i1}) d^{t+\eta \sqrt{P}}(a_{i2}, a'_{i2}) + 
  \tfrac{1}{m} \sum_{i\in T^\compl}
  d^{t+\eta \sqrt{P}}(a_{i1}, a'_{i1}) d^{t+\eta \sqrt{P}}(a_{i2}, a'_{i2})\\  
&\leq \ts \tfrac{1}{m} \sum_{i=1}^m  d^{t}(a_{i1} + r_{i1}, a'_{i1} + r'_{i1}) d^{t}(a_{i2} + r_{i2},
  a'_{i2} + r'_{i2})\ +\ \tfrac{1}{m} \sum_{i\in T^\compl} R_i,
\end{align*}  
with 
\begin{align*}
R_i&:= d^{t+\eta \sqrt{P}-\rho_{i1}}(a_{i1} + r_{i1}, a'_{i1} + r'_{i1}) d^{t+\eta
     \sqrt{P}-\rho_{i2}}(a_{i2} + r_{i2},
  a'_{i2} + r'_{i2})\\
&\qquad - d^{t}(a_{i1} + r_{i1}, a'_{i1} + r'_{i1}) d^{t}(a_{i2} + r_{i2},
  a'_{i2} + r'_{i2})\\
&= d^{t+\eta \sqrt{P}-\rho_{i1}}(a_{i1} + r_{i1}, a'_{i1} + r'_{i1}) \big(d^{t+\eta
     \sqrt{P}-\rho_{i2}}(a_{i2} + r_{i2},
  a'_{i2} + r'_{i2}) - d^{t}(a_{i2} + r_{i2},
  a'_{i2} + r'_{i2})\big)\\
&\quad + d^{t}(a_{i2} + r_{i2},
  a'_{i2} + r'_{i2}) \big( d^{t+\eta \sqrt{P}-\rho_{i1}}(a_{i1} +
  r_{i1}, a'_{i1} + r'_{i1}) - d^{t}(a_{i1} + r_{i1}, a'_{i1} + r'_{i1})\big).
\end{align*}
Using~\eqref{eq:diff-dt-ds} and~\eqref{eq:diff-dt-dist}, we have for some $c>0$ and $i \in T^\compl$,
\begin{align*}
d^{t}(a_{i2} + r_{i2},
  a'_{i2} + r'_{i2})&\leq |a_{i2} - a'_{i2} + r_{i2} - r'_{i2}| +
                      4\delta + 4|t| \lesssim |a_{i2} - a'_{i2}| +
                      \rho_{i2} + \delta + |t|,\\
d^{t+\eta \sqrt{P}-\rho_{i1}}(a_{i1} + r_{i1}, a'_{i1} +
  r'_{i1})&\lesssim |a_{i1} - a'_{i1}| +
                      \rho_{i1} + \delta + |t| - \eta \sqrt P \leq |a_{i1} - a'_{i1}| +
                      \rho_{i1} + \delta + |t|,
\end{align*}
since $\rho_{ij} \geq \eta \sqrt P$, while 
\begin{align*}
\big| d^{t+\eta \sqrt{P} - \rho_{ij}}(a_{ij} + r_{ij}, a'_{ij} + r'_{ij}) - d^{t}(a_{ij} +
  r_{ij}, a'_{ij} + r'_{ij})\big|&\leq c\rho_{ij} + c(\delta - \eta
                                   \sqrt{P})\ \lesssim \rho_{ij} + \delta. 
\end{align*}
Therefore 
\begin{align*}
R_i&\lesssim (|a_{i1} - a'_{i1}| +
                      \rho_{i1} + \delta + |t|)(\rho_{i2} + \delta) + (|a_{i2} - a'_{i2}| +
                      \rho_{i2} + \delta + |t|)(\rho_{i1} + \delta).
\end{align*}
In the last bound, the first term of the RHS can be bounded as
\begin{align*}
&\ts \sum_{i\in
  T^\compl} (|a_{i1} - a'_{i1}| +
                      \rho_{i1} + \delta + |t|)(\rho_{i2} + \delta)\\
&\ts\lesssim \|\bs a_1 - \bs
                 a'_1\|\|(\bs \rho_2)_{T^\compl}\| + \|(\bs \rho_1)_{T^\compl}\|_2 \|(\bs \rho_2)_{T^\compl}\|_2 + (\delta + |t|)
                 \|(\bs \rho_2)_{T^\compl}\|_1\\
&\qquad\qquad + \delta \|(\bs a_1 - \bs
                 a'_1)_{T^\compl}\|_1 +  \delta \|(\bs \rho_1)_{T^\compl}\|_1 + \delta (\delta + |t|)\, |T^\compl|\\
&\ts\lesssim (\eta + \tfrac{\delta}{\sqrt P}) \sqrt m \|\bs a_1 - \bs
                 a'_1\| + (\delta + |t| + \eta)
  m \eta + \delta (\delta + |t|)\, \tfrac{2m}{P},
\end{align*}
where we used the crude bounds $\max(\|(\bs
\rho_1)_{T^\compl}\|_1,\|(\bs \rho_2)_{T^\compl}\|_1) \leq \sqrt m
\max(\|\bs \rho_1\|_2,\|\bs \rho_2\|_2) \leq 2 m \eta$ and $\max(\|(\bs \rho_1)_{T^\compl}\|_2,\|(\bs \rho_2)_{T^\compl}\|_2) \leq 2 \eta \sqrt m$. 
From a similar development on the second term of the bound over $R_i$
above, we find then
\begin{align*}
&\ts \sum_{i\in
  T^\compl} R_i\\
&\ts\lesssim (\eta + \tfrac{\delta}{\sqrt P}) \sqrt m (\|\bs a_1 - \bs
                 a'_1\| + \|\bs a_2 - \bs
                 a'_2\|) + (\delta + |t| + \eta)
  m \eta + \delta (\delta + |t|)\, \tfrac{2m}{P}\\
&\ts\lesssim (\eta + \tfrac{\delta}{\sqrt P}) \sqrt m \|\bs A - \bs
                 A'\|_F + (\delta + |t| + \eta)
  m \eta + \delta (\delta + |t|)\, \tfrac{2m}{P}.
\end{align*}

\noindent This provides finally 
\begin{align*}
\cl D^{t+\eta \sqrt{P}}(\bs A,\bs A') - \cl D^{t}(\bs
                                       A + \bs R,\bs A' + \bs
  R')&\ts \lesssim (\eta + \tfrac{\delta}{\sqrt P}) \tinv{\sqrt m} \|\bs A - \bs
                 A'\|_F + (\delta + |t| + \eta)
  \eta + \delta (\delta + |t|)\, \tfrac{2}{P}.
\end{align*}

The upper bound is obtained similarly by observing that
\begin{align*}
&\cl D^{t-\eta \sqrt{P}}(\bs A,\bs A') =\ts \tfrac{1}{m} \sum_{i=1}^m
  d^{t-\eta \sqrt{P}}(a_{i1}, a'_{i1}) d^{t-\eta \sqrt{P}}(a_{i2}, a'_{i2})\\
&= \ts \tfrac{1}{m} \sum_{i\in T}
  d^{t-\eta \sqrt{P}}(a_{i1}, a'_{i1}) d^{t-\eta \sqrt{P}}(a_{i2}, a'_{i2}) + 
  \tfrac{1}{m} \sum_{i\in T^\compl}
  d^{t-\eta \sqrt{P}}(a_{i1}, a'_{i1}) d^{t-\eta \sqrt{P}}(a_{i2}, a'_{i2})\\  
&\geq \ts \tfrac{1}{m} \sum_{i=1}^m
d^{t}(a_{i1} + r_{i1}, a'_{i1} + r'_{i1}) d^{t}(a_{i2}+r_{i2}, a'_{i2}+r'_{i2})  
\ -\ \tfrac{1}{m} |\sum_{i\in
  T^\compl} R'_i|,
\end{align*} 
with 
\begin{align*}
R'_i&:= d^{t-\eta \sqrt{P}+\rho_{i1}}(a_{i1} + r_{i1}, a'_{i1} + r'_{i1}) d^{t-\eta
     \sqrt{P}+\rho_{i2}}(a_{i2} + r_{i2},
  a'_{i2} + r'_{i2})\\
&\qquad - d^{t}(a_{i1} + r_{i1}, a'_{i1} + r'_{i1}) d^{t}(a_{i2} + r_{i2},
  a'_{i2} + r'_{i2})\\
&= d^{t-\eta \sqrt{P}+\rho_{i1}}(a_{i1} + r_{i1}, a'_{i1} + r'_{i1}) \big(d^{t-\eta
     \sqrt{P}+\rho_{i2}}(a_{i2} + r_{i2},
  a'_{i2} + r'_{i2}) - d^{t}(a_{i2} + r_{i2},
  a'_{i2} + r'_{i2})\big)\\
&\quad + d^{t}(a_{i2} + r_{i2},
  a'_{i2} + r'_{i2}) \big( d^{t-\eta \sqrt{P}+\rho_{i1}}(a_{i1} +
  r_{i1}, a'_{i1} + r'_{i1}) - d^{t}(a_{i1} + r_{i1}, a'_{i1} + r'_{i1})\big).
\end{align*}
which has the same magnitude upper bound than the one found for $R_i$ above.

\section{Gaussian mean width of rank-$r$ and $s$-joint sparse matrices}
\label{sec:gmw-rank-jsparse}

This small appendix provides an upper-bound on the Gaussian mean width of $\cl R_{r,s} \cap \bb B_{\ell_F}^{n_1 \times n_2}$, with 
$$
\cl R_{r,s} := \{\bs U \in \bb R^{n_1 \times n_2}: \rank \bs U = r,\ |\supp_{\!\rm r}(\bs U)| \leq s\},
$$
where $\supp_{\!\rm r}(\bs U) := \{i \in [n_1]: \|\bs U_{i,:}\| \neq 0\}$ and $\bs U_{i,:}$ is the $i^{\rm th}$ row of $\bs U$. Notice that
$$
\ts \cl R_{r,s} = \bigcup_{T \subset [n_1],\,|T|=s} \cl R_{r, T},
$$
with $\cl R_{r, T} := \{\bs U \in \bb R^{n_1 \times n_2}: \rank \bs U
= r,\ \supp_{\!\rm r}(\bs U) = T\}$. Moreover, as given in
Tab.~\ref{tab:Kolmog-bounds}, we have $w(\cl R_{r, T} \cap \bb
B_{\ell_F}^{n_1 \times n_2})^2 \lesssim r (s + n_2)$ since $\cl R_{r,
  T}$ is isomorphic to rank-$r$ matrices in $\bb B_{\ell_F}^{s \times
  n_2}$. Therefore, since there are ${n_1 \choose s}$ $s$-cardinality $T$,
Lemma~\ref{lem:gaussian-mean-width-union-set} given below provides 
$$
\ts w(\cl R_{r,s})^2\ \lesssim r (s + n_2) + s \log \frac{n_1}{s}.
$$

\begin{lemma}
\label{lem:gaussian-mean-width-union-set}
If $\cl K = \cup_{i=1} ^T \cl K_i \subset \bb R^n$ with $\max_i w(\cl
K_i) \leq W$ and $\max_i \cl H(\cl K_i,\eta) \leq H(\eta)$ for all
$\eta > 0$,
then $w(\cl K)^2\ \lesssim\ W^2 + \log T$ and $\cl H(\cl K, \eta)
\lesssim H(\eta) + \log T$. 
\end{lemma}
\begin{proof}
We first notice that $\cl N(\cl K, \eta) \lesssim \sum_i\,\cl N(\cl K_i, \eta)$,
therefore $\cl H(\cl K, \eta) \lesssim \log T + H(\eta)$. Second, by definition, $w(\cl K) = \bb E \sup_{\bs x \in \cl K} |\scp{\bs
  g}{\bs x}| = \bb E \max_i \sup_{\bs x \in \cl K_i} |\scp{\bs
  g}{\bs x}|$, with $\bs g \sim \cl N^n(0,1)$. From \cite[Eq. (2) with $k=1$]{bandeira201518}, which
relies on the Lipschitz continuity of the map $\bs u \to L(\bs u) =
\sup_{\bs x \in \cl K_i} |\scp{\bs u}{\bs x}|$ with $\bb E L(\bs g) =
w(\cl K_i)$,
$$
\ts \bb P[\,\sup_{\bs x \in \cl K_i} |\scp{\bs g}{\bs x}| \geq
c + W + t\,]\ \leq\ \bb P[\,\sup_{\bs x \in \cl K_i} |\scp{\bs g}{\bs x}| \geq
c + w(\cl K_i) + t\,]\ \leq\ e^{-\frac{t^2}{2}}.
$$ 
with $c := \sqrt{{2}/{\pi}}$. Therefore, by union bound $\bb P[\,\sup_{\bs x \in \cl K} |\scp{\bs g}{\bs x}| \geq
c + W + t\,] \leq e^{-\frac{t^2}{2} + \log T}$, which implies
$$
\ts \bb P[\,\sup_{\bs x \in \cl K} |\scp{\bs g}{\bs x}| \geq
c + W + t + \sqrt{2\log T}\,]\ \leq\ e^{-\frac{t^2}{2}}.
$$
Consequently, for $a = c + W + \sqrt{2\log T}$,
\begin{align*}
w(\cl K)&\ts = \int_0^{+\infty} \bb P[\sup_{\bs x \in \cl K} |\scp{\bs g}{\bs x}| \geq
t]\,\ud t\\
&\ts = \int_0^{a} \bb P[\sup_{\bs x \in \cl K} |\scp{\bs g}{\bs x}| \geq
t]\,\ud t\  +\ \int_0^{+\infty} \bb P[\sup_{\bs x \in \cl K} |\scp{\bs g}{\bs x}| \geq
a + t]\,\ud t\\
&\ts \leq a + \int_0^{+\infty} e^{-\frac{t^2}{2}} = a +
  \gone\ \lesssim\ W + \sqrt{\log T}.
\end{align*}
\end{proof}

\end{document}